\documentclass[twocolumn]{autart}    

\usepackage{amsmath,amssymb,amsfonts,bm}
\usepackage{graphicx}
\usepackage{algorithm}
\usepackage{algorithmicx,algpseudocode}
\usepackage{cite}
\usepackage{color}
\let\theoremstyle\relax

\usepackage{amsthm}
\usepackage{cuted}
\usepackage[mathscr]{eucal}
\theoremstyle{plain}
\newtheorem{theorem}{Theorem}[section]
\newtheorem{lemma}[theorem]{Lemma}
\newtheorem{assumption}[theorem]{Assumption}
\newtheorem{corollary}[theorem]{Corollary}
\newtheorem{proposition}[theorem]{Proposition}
\theoremstyle{definition}

\theoremstyle{remark}
\newtheorem*{remark}{Remark}
\usepackage{wrapfig}

\usepackage{wrapfig}
\usepackage[T1]{fontenc}
\usepackage[utf8]{inputenc}

\begin{document}

\begin{frontmatter}

\title{Topological Feasibility Guarantees for Differentiable Predictive Control} 

\author[Guangyu,Guangyu1]{Guangyu Wu}\ead{gwu1@alumni.nd.edu},    
\author[{Jan1}]{J\'an Drgo\v na}\ead{jdrgona1@jh.edu}              

\thanks{Part of the conceptual work was initiated while the author was with Nanyang Technological University.}
\address[Guangyu]{Department of Mathematical Sciences, Chalmers University of Technology, Sweden (Incoming)} 
\address[Guangyu1]{College of Computing and Data Science, Nanyang Technological University, Singapore}  
\address[Jan1]{Department of Civil and Systems Engineering, Johns Hopkins University, USA}
          

\begin{abstract}
{Differentiable predictive control (DPC), a self-supervised learning approach for approximating explicit model predictive control (MPC) policies, offers significant computational advantages over online optimization-based MPC. However, feasibility guarantees, a core requirement for safe control, are currently provided either probabilistically or via online safety filters. The lack of rigorous feasibility guarantees for offline policy optimization remains an open problem.
This paper establishes deterministic feasibility guarantees for DPC using a novel topological analysis of the induced reachable safe set, without requiring online safety filters. By exploiting the inherent model-based nature of DPC, in which differentiable system dynamics are embedded directly into the computational graph, we analyze the properties of the learned control policies and the corresponding system states from topological and geometric perspectives. Inspired by our theoretical analysis, we propose a novel self-supervised offline policy learning strategy that utilizes a proxy loss with Control Barrier Functions (CBFs). Crucially, these properties not only significantly improve policy training but also enable the derivation of strict, deterministic feasibility guarantees from a finite number of training samples. Extensive closed-loop simulations validate our theoretical findings, demonstrating that the empirical constraint violations monotonically decrease to zero as the training sample size increases. Ultimately, this work illustrates that DPC policy optimization yields formal safety certificates that are structurally unattainable with conventional black-box methods, e.g., reinforcement learning (RL) or supervised learning-based approximate MPC, thereby providing a new perspective on feasibility guarantees in learning-based control.}
\end{abstract}

\maketitle

\end{frontmatter}

\section{Introduction}

Model Predictive Control (MPC) computes control actions by repeatedly solving a finite-horizon constrained optimal control problem using a predictive model of the system~\cite{mayne2000constrained,rawlings2017model,borrelli2017predictive,Grune2017}. By explicitly enforcing state and input constraints together with suitable terminal ingredients, MPC provides recursive feasibility and closed-loop stability guarantees under standard assumptions for a broad class of systems~\cite{camacho2004model,mayne2014model,Mayne2005,Lazar2006,ZHANG20168,MAIWORM201550,mesbah2016stochastic,Kohler2021}.

There are three main paradigms for solving MPC problems: implicit MPC~\cite{DIEHL2002577,Diehl2009,Domahidi2012,Wu2025}, explicit MPC~\cite{bemporad2002explicit,zeilinger2011real,kvasnica2011clipping}, and approximate MPC~\cite{hertneck2018learning,Karg2020,Chen2018}. Implicit MPC computes the optimal control sequence online by solving an optimization problem at every sampling instant, providing exact solutions at the expense of substantial computational effort. Explicit MPC shifts this computation offline via multi-parametric programming~\cite{Herceg2013,Bemporad2000}, yielding an explicit policy that can be evaluated efficiently online. However, the number of critical regions grows exponentially with the number of constraints, parameters, and prediction horizon, limiting scalability~\cite{Alessio2009}.
Approximate MPC addresses this limitation by replacing the exact explicit solution with a compact function approximator. Early approaches employed piecewise-affine representations~\cite{bemporad2011ultra,Jones2010}, while more recent methods leverage kernel-based or neural-network approximators~\cite{hertneck2018learning,tokmak2025automatic}. These approaches are typically trained offline from optimal control trajectories generated by implicit MPC and evaluated online through a single function evaluation, enabling orders-of-magnitude faster inference. This computational efficiency, however, comes at the expense of theoretical guarantees. While probabilistic guarantees have been established for neural-network approximations~\cite{hertneck2018learning}, deterministic guarantees remain difficult to obtain. Recent kernel-based methods~\cite{tokmak2025automatic} address this issue by deriving certified approximation error bounds that recover deterministic closed-loop guarantees. However, these methods suffer from the curse of dimensionality and scale poorly to high-dimensional nonlinear systems. Consequently, achieving deterministic feasibility guarantees without sacrificing scalability remains an open problem.

While existing approximate MPC methods rely on supervised learning from optimal MPC trajectories, differentiable predictive control (DPC)~\cite{drgovna2024learning,DRGONA202280} adopts a fundamentally different paradigm. By embedding the system dynamics directly into the learning architecture, DPC optimizes an explicit control policy by differentiating the MPC objective through closed-loop simulations, eliminating the need for expensive offline MPC solutions. This self-supervised learning methodology significantly improves scalability and flexibility, but providing formal feasibility guarantees for the resulting policies still remains a major challenge.

Existing safe learning-based control methods~\cite{hewing2020learning,Brunke2022,Drgona2025} predominantly provide probabilistic rather than deterministic feasibility guarantees through statistical confidence bounds, chance-constrained MPC, or probabilistically constructed safe sets in reinforcement learning~\cite{lederer2025risk,mestres2024feasibility,berkenkamp2017safe,garcia2015comprehensive}. 
To obtain deterministic safety guarantees, existing approaches generally follow one of two paradigms. The first, exemplified by learning-based MPC~\cite{aswani2013provably,Hewing2020,koller2018learning,Rosolia2018,Berberich2021,Coulson2019,KORDA2018149}, learns the system model while retaining an online MPC optimization step to preserve the recursive feasibility and stability guarantees of implicit MPC. 
The second augments learned explicit policies with online safety mechanisms~\cite{Wabersich2023}, including Control Barrier Functions (CBFs)~\cite{ames2016control,ames2019control,Lindemann2019,Xiao2023,Yang2023,DIDIER2026113009}, Hamilton-Jacobi (HJ) reachability~\cite{Fisac2019,Gillula2010,Bansal2017}, predictive safety filters~\cite{WABERSICH2021109597,leeman23a,cortez2024psf}, convex optimization-based projection layers~\cite{donti2021enforcing}, and online feasibility verification and fallback policies~\cite{Hose2025}. While both paradigms provide formal safety guarantees under suitable assumptions, they do so either by solving an optimization problem online or by modifying the learned policy online. The former sacrifices the computational advantages of explicit control policies, while the latter may compromise trajectory-level performance when corrective actions are applied myopically.

From this perspective, the central limitation is that existing methods repair infeasible policies \emph{a posteriori} rather than constructing policies that are feasible by design. 
Fundamentally, this reflects the absence of an intrinsic characterization of the feasible control manifold, requiring online corrective actions to recover safety or feasibility~\cite{chen2020guaranteed,cortez2022differentiable,xiao2023barriernet,ding2024online,he2025state,didier2026approximate,liu2025robust,Robey2020}.

These current limitations raise the fundamental question of whether learning-based controllers can be trained entirely offline while providing deterministic feasibility guarantees comparable to those of optimization-based MPC, without resorting to online optimization or safety filtering. Such a capability would substantially improve the reliability and practical deployment of learning-based control in time- and safety-critical applications.

To answer this question, we establish a theoretical framework for constructing DPC policies with strict deterministic closed-loop feasibility guarantees through offline training alone. Our approach analytically characterizes the topological structure of the feasible control manifold, ensuring that the learned explicit policy satisfies feasibility by construction and therefore requires neither \textit{a posteriori} online safety filters nor optimization-based corrective actions.

The main contributions of this paper are as follows. 
We develop a novel topological framework for analyzing the feasibility of DPC policies, establishing deterministic closed-loop feasibility guarantees from a finite number of offline training samples under mild assumptions. Building on this analysis, we develop a self-supervised offline learning strategy based on a Control Barrier Function (CBF) proxy loss that constructs neural control policies satisfying feasibility by design, eliminating the need for online safety filters or optimization-based corrective actions. Finally, we demonstrate the effectiveness of the proposed framework on several nonlinear control benchmarks, with theoretical guarantees corroborated by empirical results showing that constraint violations vanish as the number of training samples increases.

The remainder of the paper is organized as follows. Section~2 defines the parametric optimal control problem formulation and reviews DPC methodology. Section~3 presents the proposed topological feasibility analysis, and Section~4 introduces the corresponding CBF-guided offline training strategy. Section~5 discusses the proposed framework, while Section~6 presents numerical results.

\section{Differentiable Predictive Control}

We consider the continuous-time process governed by
\begin{equation}
    \d\boldsymbol{x}_t 
    = \mathbf{f}(\boldsymbol{x}_t, \boldsymbol{u}_t)\,\d t
\label{SDE}
\end{equation}
where $\boldsymbol{x}_t \in \mathbb{R}^{n_{x}}$ denotes the state vector and $\boldsymbol{u}_t \in \mathbb{R}^{n_{u}}$ the control input. 
The function $\mathbf{f}:\mathbb{R}^{n_{x}} \times \mathbb{R}^{n_{u}} \rightarrow \mathbb{R}^{n_{x}}$ represents the deterministic drift term of the dynamics. The formulation \eqref{SDE} is a general form of the ODE where $\mathbf{f}$ is either linear or nonlinear.

To obtain a discrete-time representation, the continuous-time dynamics are numerically integrated over each sampling interval. Consider the time interval $[0, T]$. The deterministic drift term is discretized using e.g., a Runge Kutta scheme, yielding the nominal discrete-time mapping 
$\boldsymbol{\mathcal{F}}:\mathbb{R}^{n_{x}} \times \mathbb{R}^{n_{u}} \rightarrow \mathbb{R}^{n_{x}}$. Denote $N$ as the prediction horizon, and $k = \frac{T}{N}$ as a discrete
time step. The resulting discrete-time state-space model is then expressed as
\begin{equation}
    \boldsymbol{x}_{k+1} = \boldsymbol{\mathcal{F}}(\boldsymbol{x}_k, \boldsymbol{u}_k),
\label{syseq}
\end{equation}
where $\boldsymbol{x}_k \in \mathbb{R}^{n_{x}}$ denotes the system state at time step $k$. 

We first review the differentiable predictive control (DPC) method in \cite{drgovna2024learning} for the ODE system governed by \eqref{syseq}. In contrast to implicit MPC, which optimizes the control input trajectories online for a given single parametric instance, the DPC aims to optimize an explicit policy $
\pi: \left(\boldsymbol{x}_{k}, \boldsymbol{\xi}_{k} \right) \rightarrow \boldsymbol{u}_{k}
$ that maps the current state and parameters to the control action.

We give a parametric formulation for the infinite-dimensional mapping $\pi$, referred to as a neural control policy, of which the parameters are denoted as $\mathbf{W}$. The problem parameters, reference state, state constraint and input constraint are denoted as $\boldsymbol{\xi}_{k}, \boldsymbol{r}_{k}, \mathbf{p}_{\mathbf{h}_{k}}, \mathbf{p}_{\mathbf{g}_{k}}$ respectively. The differentiable predictive control can be formulated as
\begin{subequations}\label{eq:opt_problem}
\begin{align}
&\min_{\mathbf{W}} \mathcal{L}_{NC}(\boldsymbol{x}_{k}, \boldsymbol{u}_k, \mathbf{r}_k) \quad 
&   \label{eq:opt_problem_cost} \\[3pt]
\text{s.t.} \quad 
& \boldsymbol{x}_{k+1} = \boldsymbol{\mathcal{F}}(\boldsymbol{x}_k, \boldsymbol{u}_k), \ k \in \mathbb{N}_0^{N-1} \label{eq:opt_problem_dyn} \\
& \boldsymbol{u}_{k} = \pi_{\mathbf{W}}(\boldsymbol{x}_{k}, \boldsymbol{\xi}_k) 
\label{eq:opt_problem_policy} \\
& h(\boldsymbol{x}_k, \mathbf{p}_{\mathbf{h}_k}) \geqslant 0
\label{eq:opt_problem_state} \\
& g(\boldsymbol{u}_k, \mathbf{p}_{\mathbf{g}_k}) \geqslant 0
\label{eq:opt_problem_input} \\
& \boldsymbol{x}_0 \in \mathbb{R}^{n_x},
\boldsymbol{\xi}_k = \{\mathbf{r}_k, \mathbf{p}_{\mathbf{h}_k},
\mathbf{p}_{\mathbf{g}_k}, \theta\} \in \mathbb{R}^{n_\xi},
\label{eq:opt_problem_vars}
\end{align}
\end{subequations}
where the cost function reads
\begin{equation}
\begin{aligned}
& \mathcal{L}_{NC}(\mathbf{W})\\= & \frac{1}{m N} \sum_{i=1}^m \sum_{k=0}^{N-1}\left(\alpha_l \ell\left(\boldsymbol{x}_k^{(i)}, \boldsymbol{u}_k^{(i)}, \boldsymbol{r}_k^{(i)}\right) \right.\\
+ & \left.\alpha_h p_x\left(h\left(\boldsymbol{x}_k^{(i)}, \mathbf{p}_{\mathbf{h}_{k}}^{(i)}\right)\right) + \alpha_g p_u\left(g\left(\boldsymbol{u}_k^{(i)}, \mathbf{p}_{\mathbf{g}_{k}}^{(i)}\right) \right) \right.\\
+ & \left.\alpha_N p_N\left(\boldsymbol{x}_N^{(i)}\right)\right)
\label{Lnc}
\end{aligned}
\end{equation}
with $m$ being the number of training data samples, and $\alpha_{l}, \alpha_{h}, \alpha_{g}, \alpha_{N}$ being the weighting coefficients. The first term in the RHS of \eqref{Lnc} represents the performance metric, which can be defined, e.g. as a reference tracking term \cite{drgovna2024learning}, 
$
\ell\left(\boldsymbol{x}^{(i)}_k, \boldsymbol{u}^{(i)}_k, \boldsymbol{r}^{(i)}_k\right)=\left\|\boldsymbol{x}^{(i)}_k-\boldsymbol{r}^{(i)}_k\right\|_2^2+\left\|\boldsymbol{u}^{(i)}_k\right\|_2^2
$.
The terminal penalty term \(p_N(\boldsymbol{x}^{(i)}_N)\) in the $\mathcal{L}_{NC}$ formulation is a standard MPC component designed to ensure recursive feasibility and closed-loop stability. It can be synthesized using established methods \cite{chen1998quasi, mayne2000survey, fagiano2013generalized, kohler2019nonlinear}.
The second, third and fourth terms in the RHS of \eqref{Lnc} represent the state, input and terminal constraints respectively. The state and input constraints can be written in the following general forms
$ p_x\left(h\left(\boldsymbol{x}^{(i)}_k, \mathbf{p}_{\mathbf{h}_k}\right)\right)=\frac{1}{n_h}\mu_{x}\left(h\left(\boldsymbol{x}^{(i)}_k, \mathbf{p}_{\mathbf{h}_k}\right)\right)$ and $p_u\left(g\left(\boldsymbol{u}^{(i)}_k, \mathbf{p}_{\mathbf{g}_k}\right)\right)=\frac{1}{n_g}\mu_{u}\left(g\left(\boldsymbol{u}^{(i)}_k, \mathbf{p}_{\mathbf{g}_k}\right)\right)$
where $\mu_{x}, \mu_{u}:\mathbb{R} \mapsto \mathbb{R}$ are the functions that aim to penalize violations of the constraints \eqref{eq:opt_problem_state} and \eqref{eq:opt_problem_input}.

There are typically two types of penalty functions. The first type has a less strict form. For example, the ReLU function $\mu(x)=\max(0,x)$ \cite{drgovna2024learning} and the GeLU function $\mu(x) = x\frac{1}{\sqrt{2 \pi}} \int_{-\infty}^x e^{-t^2 / 2} d t$ \cite{lee2023mathematical}. Because these functions are defined on the whole real line \(\mathbb{R}\) rather than exactly on the feasible set determined by the constraints, they implement \emph{soft} penalties rather than hard constraints—the larger the deviation, the more the penalty. The second type, in contrast, is strictly defined on the domain confined by the constraints. A
canonical example is the log barrier \cite{boyd2004convex}. For example, to let the system state satisfy the constraint \eqref{eq:opt_problem_state}, the function can be chosen as $\mu(x) = - \log (-x+\epsilon)$, where $\epsilon>0$ is a sufficiently small margin. Each type has clear advantages and disadvantages. To improve the feasibility of the learning-based control policy, we propose to employ both types of penalties in the following parts of the paper.

A fundamental consequence of adopting differentiable closed-loop dynamics, stage costs, constraint mappings, and terminal penalty functions is that the resulting optimization problem becomes amenable to gradient-based training via automatic differentiation, namely, backpropagation through time (BPTT)~\cite{Werbos1990,puskorius1994truncated}. More precisely, upon representing problem~(1) as a computational graph, repeated application of the chain rule yields the gradient of the DPC objective $\mathcal{L}_{NC}$ with respect to the policy parameters $\mathbf{W}$.

For illustration and notational simplicity, consider the scalar-input and single-stage setting $m=1$ and $N=1$ and write $\boldsymbol{x}_{0}^{(1)}, \boldsymbol{u}_{0}^{(1)}, \boldsymbol{r}_{0}^{(1)}$ as $\boldsymbol{x}, \boldsymbol{u}, \boldsymbol{r}$. The corresponding gradient expression is given by
\begin{equation}
\begin{aligned}
& \nabla_{\mathbf{W}} \mathcal{L}_{NC} = \frac{\partial \ell(\boldsymbol{x}, \boldsymbol{u}, \boldsymbol{r})}{\partial \mathbf{W}}
+
\frac{\partial p_x\!\left(h\!\left(\boldsymbol{x}, \boldsymbol{p}_{\boldsymbol{h}}\right)\right)}{\partial \mathbf{W}}\\
+ &
\frac{\partial p_u\!\left(g\!\left(\boldsymbol{u}, \boldsymbol{p}_{\boldsymbol{g}}\right)\right)}{\partial \mathbf{W}}
+
\frac{\partial p_N(\boldsymbol{x}_N)}{\partial \mathbf{W}}
\\[0.3em]
= &
\frac{\partial \ell(\boldsymbol{x}, \boldsymbol{u}, \boldsymbol{r})}{\partial \boldsymbol{x}}
\frac{\partial \boldsymbol{x}}{\partial \boldsymbol{u}}
\frac{\partial \boldsymbol{u}}{\partial \mathbf{W}}
+
\frac{\partial \ell(\boldsymbol{x}, \boldsymbol{u}, \boldsymbol{r})}{\partial \boldsymbol{u}}
\frac{\partial \boldsymbol{u}}{\partial \mathbf{W}}
\\
+ &
\frac{\partial p_x\!\left(h\!\left(\boldsymbol{x}, \boldsymbol{p}_{\boldsymbol{h}}\right)\right)}{\partial \boldsymbol{x}}
\frac{\partial \boldsymbol{x}}{\partial \boldsymbol{u}}
\frac{\partial \boldsymbol{u}}{\partial \mathbf{W}}
+
\frac{\partial p_u\!\left(g\!\left(\boldsymbol{u}, \boldsymbol{p}_{\boldsymbol{g}}\right)\right)}{\partial \boldsymbol{u}}
\frac{\partial \boldsymbol{u}}{\partial \mathbf{W}}
\\
+ &
\frac{\partial p_N(\boldsymbol{x}_N)}{\partial \boldsymbol{x}}
\frac{\partial \boldsymbol{x}}{\partial \boldsymbol{u}}
\frac{\partial \boldsymbol{u}}{\partial \mathbf{W}} .
\end{aligned}
\label{eq:DPC_gradient}
\end{equation}

The explicit availability of the gradient representation in~\eqref{eq:DPC_gradient} permits the application of first-order stochastic optimization methods, including stochastic gradient descent and its adaptive variants such as AdamW~\cite{loshchilov2017decoupled}, for the numerical solution of the parametric optimal control problem~\eqref{eq:opt_problem}. In practice, these derivatives may be evaluated efficiently through modern auto-differentiation frameworks, for instance PyTorch \cite{paszke2019pytorch}.

\section{Feasibility Analyses of DPC}

The standard DPC algorithm has shown significant advantages in scalability and inference time. However, a rigorous feasibility proof of DPC remains a significant challenge. In this section, we will show that the feasibility of the DPC can be proved rigorously in topological terms, with a sketch provided in Figure \ref{fig1} to clarify the analysis. We first denote the control input and the relative system state for sample initialization $i$ at time step $k$ by the optimization problem \eqref{eq:opt_problem} to be $\left(\boldsymbol{u}^{*(i)}_{k}, \boldsymbol{x}^{*(i)}_{k}\right)$.
\begin{figure}[htbp]
\centering
\includegraphics[scale=0.6]{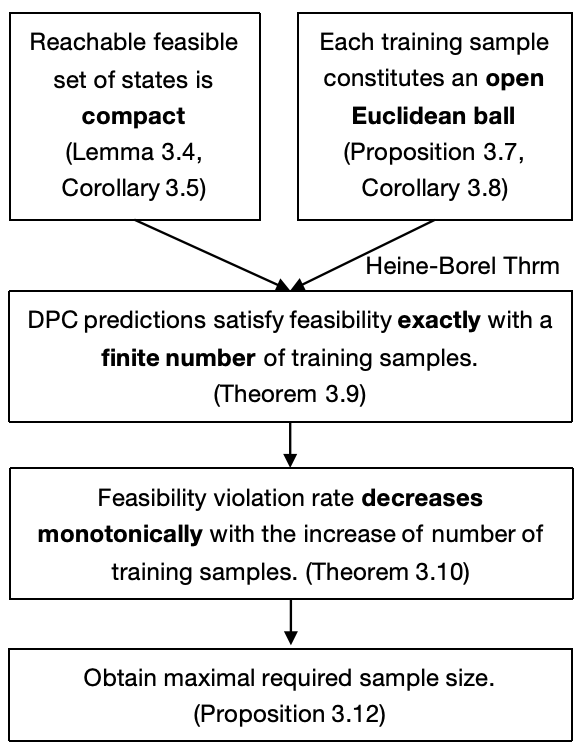}
\caption{A sketch of the proofs for feasibility analyses of DPC.}
\label{fig1}
\end{figure}

In this paper, we consider the feasibility of the DPC in the manner of CBF \cite{ames2019control}. The CBF guarantees feasibility by defining a forward invariant subset of the state space, ensuring that if a system starts within a feasible region, the control law will always keep it there. For our discrete-time setting, we define the discrete-time zeroing CBF update condition
\begin{equation}
\begin{aligned}
& r_{k+1, l}(\boldsymbol{x}^{(i)}_k)
:= \\
& h_{l}\!\left(\boldsymbol{x}^{(i)}_{k+1}, \mathbf{p}_{\mathbf{h}_{k+1}}\right)
-(1-\eta)h_{l}\!\left(\boldsymbol{x}^{(i)}_{k}, \mathbf{p}_{\mathbf{h}_{k}}\right) > 0
\label{eq:CBF-update}
\end{aligned}
\end{equation}
for each safety inequality $h_{l}$, where $\eta\in(0,1]$. Crucially, we retain the exact identity mapping for $h_{l}(\boldsymbol{x})$ to preserve the original geometry of the safe set (defined by $h_{l}(\boldsymbol{x})\geqslant 0$). Furthermore, the following results apply to each constraint $l$; for brevity, the subscript $l$ is omitted where it does not cause confusion.

Let all $\boldsymbol{x}_{k}, 0 \leqslant k \leqslant N$ in this section refer to the system state passed through the \emph{state} safety-proxy loss \eqref{eq:safety-proxy-loss}. We first have the following assumption.

\begin{assumption}[Lipschitz assumption of the safety set]
\label{assmp:41}
Denote the safety set 
$$
\mathcal{X}(k):=\left\{\boldsymbol{x}_{k} \in \mathbb{R}^{n_{x}}: h(\boldsymbol{x}_k, \mathbf{p}_{\mathbf{h}_k}) \geqslant 0\right\}
$$
and
$$
\mathcal{W}(k):=\left\{\boldsymbol{u}_{k} \in \mathbb{R}^{n_{u}}: g(\boldsymbol{x}_k, \mathbf{p}_{\mathbf{g}_k}) \geqslant 0\right\}.
$$
The function $\boldsymbol{\mathcal{F}}: \mathbb{R}^{n_{x}} \times \mathbb{R}^{n_{u}} \rightarrow \mathbb{R}^{n_{x}}$ is locally Lipschitz continuous on compact sets $\mathscr{X} \subset \mathbb{R}^{n_{x}}$ and $\mathscr{W} \subset \mathbb{R}^{n_{u}}$, for which $\mathcal{X}(k) \subset \mathscr{X}, \mathcal{W}(k) \subset \mathscr{W}, \forall k \in \mathbb{N}$, with Lipschitz coefficients $L_{\boldsymbol{\mathcal{F}},x}, L_{\boldsymbol{\mathcal{F}},u} \in \mathbb{R}_{+}$, where $\mathbb{R}_{+}$ denotes the set of positive real numbers.
\end{assumption}

\begin{assumption}[CBF condition satisfaction]\label{assmp:40}
For any initialization of $\boldsymbol{x}_{0}$ in the feasible safety set, there always exists a feasible control trajectory satisfying the CBF condition \eqref{eq:CBF-update}.
\end{assumption}
\begin{remark}
We will first concentrate on the feasibility proofs by this assumption and will focus on algorithm design for achieving this condition in the subsequent section.
\end{remark}

\begin{assumption}[Lipschitz continuity of $h$]
\label{assmp:42}
There exists a neighborhood of the point
$(\boldsymbol{x}_k^\star,\mathbf{p}_{\mathbf{h}_k})$
and constants $L_{h,x}>0$ such that for all 
$(\boldsymbol{x}_{k,1},\mathbf{p}_{\mathbf{h}_k})$ and $(\boldsymbol{x}_{k,2},\mathbf{p}_{\mathbf{h}_k})$ in this neighborhood,
\[
\begin{aligned}
|h(\boldsymbol{x}_{k,1},\mathbf{p}_{\mathbf{h}_k}))-h(\boldsymbol{x}_{k,2},\mathbf{p}_{\mathbf{h}_k}))|
\leqslant L_{h,x,k}\|\boldsymbol{x}_{k,1}-\boldsymbol{x}_{k,2}\|.
\end{aligned}
\]
Here, we note that $\mathbf{p}_{\mathbf{h}_k}$ are deterministic parameters that do not depend on the state $\boldsymbol{x}_k$.
\end{assumption}

\begin{assumption}[Lipschitz continuity of the neural policy]
\label{assmp:43}
Assume that the learned neural control policy $\pi_\mathbf{W}\left(\boldsymbol{x}_k, \boldsymbol{\xi}_k\right)$ is locally Lipschitz continuous with respect to the state $\boldsymbol{x}_k$ on the compact set $\mathscr{X}$, with a Lipschitz constant $L_\pi>0$. That is, $\left\|\pi_\mathbf{W}\left(\boldsymbol{x}_{k, 1}, \boldsymbol{\xi}_{k}\right)-\pi_{\mathbf{W}}\left(\boldsymbol{x}_{k, 2},\boldsymbol{\xi}_{k}\right)\right\| \leqslant L_\pi\left\|
\boldsymbol{x}_{k, 1}-\boldsymbol{x}_{k, 2}\right\|$.
\end{assumption}

\begin{remark}
    It is worth emphasizing that our feasibility guarantees do not rely on differentiability assumptions. Indeed, even if the true optimal control mapping is non-differentiable—and a differentiable function approximator (e.g., a multi-layer perceptron) is employed in the DPC, the following results remain valid. Crucially, the learned control policy need not be close to the true optimal policy in any specific norm or metric. This minimal reliance on structural assumptions distinguishes our approach from existing methods and enhances its practical applicability in real-world scenarios.
\end{remark}

Assumptions \ref{assmp:41} and \ref{assmp:42} have a similar form as the typical Lipschitz continuity for the model predictive control of ODE control problem. By this assumption, we can then prove the following lemma.

\begin{lemma}
\label{lemma:43}
Under Assumption~\ref{assmp:41}, if $\mathcal{X}(0)$ is compact, then
each subsequent set $\mathcal{X}(k)$, $1 \leqslant k \leqslant N$, is
compact.
\end{lemma}

\begin{proof}
Fix $k \in \{0,1,\dots,N\}$ and define the function
\[
\zeta_k(\boldsymbol{x}) := h\left(\boldsymbol{x}_k, \mathbf{p}_{\mathbf{h}_k}\right), \quad
\boldsymbol{x} \in \mathbb{R}^{n_x}.
\]
By construction, $h(\cdot, \mathbf{p}_{\mathbf{h}_k})$ is continuous in $\boldsymbol{x}$. Hence $\zeta_k$ is continuous on
$\mathbb{R}^{n_x}$.

The safety set at time $k$ can be written as
\[
\mathcal{X}(k)
= \{\boldsymbol{x} \in \mathbb{R}^{n_x} : \zeta_k(\boldsymbol{x}) \leqslant 0\}
= \zeta_k^{-1}((-\infty,0]).
\]
Since $(-\infty,0]$ is a closed subset of $\mathbb{R}$ and
$\zeta_k$ is continuous, the preimage
$\zeta_k^{-1}((-\infty,0])$ is closed in $\mathbb{R}^{n_x}$.

Moreover, by Assumption~\ref{assmp:41} there exists a compact set
$\mathscr{X} \subset \mathbb{R}^{n_x}$ such that
$\mathcal{X}(k) \subset \mathscr{X}$ for all $k$. Therefore
we can equivalently write
\[
\mathcal{X}(k)
= \zeta_k^{-1}((-\infty,0]) \cap \mathscr{X},
\]
which is an intersection of a closed set with a compact set.
Thus $\mathcal{X}(k)$ is closed and bounded, hence compact.

This argument applies to each $k=0,1,\dots,N$, so in
particular $\mathcal{X}(k)$ is compact for all
$1 \leqslant k \leqslant N$.
\end{proof}

\begin{remark}
This lemma proves that the nominal safety set at each time step $k$ is compact. However, the proof doesn't take into account the \textbf{reachability} of $\mathcal{X}(k)$, which means that a subset of $\mathcal{X}(K)$ is not achievable by the $h(\boldsymbol{x}_k, \mathbf{p}_{\mathbf{h}_k}) \leqslant 0$ in the previous steps $0 \leqslant k \leqslant K$. Denote the state sequence by $\boldsymbol{X} := (\boldsymbol{x}_0,\dots,\boldsymbol{x}_N)$, and we have the following corollary.
\end{remark}

\begin{corollary}
    Denote the domain of $\boldsymbol{u}_{k} \in \mathbb{R}^{n_{u}}$ satisfying \eqref{eq:opt_problem_input} to be a compact set $\mathcal{U}$. Define the reachable sets $\mathcal{R}(0) :=\mathcal{X}(0)$ and $\mathcal{R}(k+1) :=\left\{\boldsymbol{\mathcal{F}}\left(\boldsymbol{x}_k, \boldsymbol{u}_k\right) + \mathbf{W}_{k}: \boldsymbol{x}_k \in \mathcal{R}(k), \boldsymbol{u}_k \in \mathcal{U} \right\}$ for time step $k$. The reachable safety set $\mathcal{X}_{r}(k)$ can then be written as
    \begin{equation}
    \label{Xrk}
    \begin{aligned}
    & \mathcal{X}_{r}(k):= \\
    & \left\{\boldsymbol{x}_{k} \in \mathcal{R}(k): h\left(\boldsymbol{x}_k, \mathbf{p}_{\mathbf{h}_k}\right) \geqslant 0\right\}=\mathcal{R}(k) \cap \mathcal{X}(k).
    \end{aligned}
    \end{equation}
    We have $\mathcal{X}_{r}(k)$ is compact for each time step $0 \leqslant k \leqslant N$. Hence, the set of reachable safety sequences $\boldsymbol{\mathcal{X}}_{r}:= \left\{\boldsymbol{X}=(\boldsymbol{x}_0,\dots,\boldsymbol{x}_N) \middle| \boldsymbol{x}_k \in \mathcal{X}_{r}(k) \right\} \in \mathbb{R}^{(N+1)n_x}$ is compact.
\end{corollary}
\begin{proof}
If $\mathcal{R}(k)$ is compact and $\mathcal{U}$ is compact, 
then their Cartesian product $\mathcal{R}(k) \times \mathcal{U}$ is also compact.
Since the mapping
\[
(\boldsymbol{x}_{0}, \boldsymbol{u}_{0}) \mapsto \boldsymbol{\mathcal{F}}(\boldsymbol{x}_{0}, \boldsymbol{u}_{0})
\]
is continuous, it maps compact sets to compact sets. We note that $\mathcal{R}(k)$ are finite compact sets for $1 \leqslant k \leqslant N$). Therefore,
\[
\mathcal{R}(1) = \boldsymbol{\mathcal{F}}(\mathcal{R}(0), \mathcal{U})
\]
is compact. By \eqref{Xrk} and Lemma \ref{lemma:43}, $\mathcal{X}_{r}(1) = \mathcal{R}(1) \cap \mathcal{X}(1)$ is also compact. Similarly, we could prove $\mathcal{X}_{r}(k), 2 \leqslant k \leqslant N$ are compact sets. Therefore, $\boldsymbol{\mathcal{X}}_{r}$ is also compact, which completes the proof.
\end{proof}

In the previous results of this section, we have proved the topological property of the set of all feasible system state trajectories. Next, we would seek what property each single data sample has by our proposed algorithm. In the following lemma, we will first prove the Lipschitz continuity of the function $r_{k+1}(\boldsymbol{x}_k)$.

\begin{lemma}
\label{lem:r-lip}
Let Assumption~\ref{assmp:41}, Assumption~\ref{assmp:42}, and the assumption on the Lipschitz continuity of the neural policy hold. Then
$r_{k+1}(\boldsymbol{x}_k)$ is locally Lipschitz in $\boldsymbol{x}_k$
in a neighborhood of $\boldsymbol{x}_k^\star$. The Lipschitz coefficient is
\begin{equation}
\label{LipschitzCoeff}
    L := L_{h,x,k+1}(L_{\boldsymbol{\mathcal{F}},x} + L_{\boldsymbol{\mathcal{F}},u}L_{\pi})
      +(1-\eta)L_{h,x,k}
\end{equation}
\end{lemma}

\begin{proof}
Take two points $\boldsymbol{x}_{k, 1}$ and $\boldsymbol{x}_{k, 2}$ in a neighborhood
of $\boldsymbol{x}_k^\star$. Let the corresponding control inputs generated by the learned neural policy be $\boldsymbol{u}_{k, 1} = \pi_{\mathbf{W}}(\boldsymbol{x}_{k, 1}, \boldsymbol{\xi}_{k})$ and $\boldsymbol{u}_{k, 2} = \pi_{\mathbf{W}}(\boldsymbol{x}_{k, 2}, \boldsymbol{\xi}_{k})$. Then
\[
\begin{aligned}
&|r_{k+1}(\boldsymbol{x}_{k, 1})-r_{k+1}(\boldsymbol{x}_{k, 2})| \\
&\leqslant 
\bigl|h \left(\boldsymbol{\mathcal{F}}(\boldsymbol{x}_{k, 1},\boldsymbol{u}_{k, 1}),\mathbf{p}_{\mathbf{h}_{k+1}} \right)
- h \left(\boldsymbol{\mathcal{F}}(\boldsymbol{x}_{k, 2},\boldsymbol{u}_{k, 2}),\mathbf{p}_{\mathbf{h}_{k+1}} \right) \bigr| \\
&\quad
+ (1-\eta)\bigl|h(\boldsymbol{x}_{k, 1},\mathbf{p}_{\mathbf{h}_k})-h(\boldsymbol{x}_{k, 2},\mathbf{p}_{\mathbf{h}_k})\bigr|.
\end{aligned}
\]

For the first term, by the Lipschitz continuity of $h$ and
$\mathbf{p}_{\mathbf{h}_{k+1}}(\cdot)$ (Assumption~\ref{assmp:42}), the local Lipschitz continuity of
$\boldsymbol{\mathcal{F}}$ with respect to both state and control input (Assumption~\ref{assmp:41}), and the Lipschitz continuity of the policy network with constant $L_{\pi}$ (Assumption~\ref{assmp:43}), we have
\[
\begin{aligned}
&\bigl|h\left(\boldsymbol{\mathcal{F}}(\boldsymbol{x}_{k, 1},\boldsymbol{u}_{k, 1}),\mathbf{p}_{\mathbf{h}_{k+1}} \right)
- h \left(\boldsymbol{\mathcal{F}}(\boldsymbol{x}_{k, 2},\boldsymbol{u}_{k, 2}),\mathbf{p}_{\mathbf{h}_{k+1}} \right)\bigr| \\
&\leqslant L_{h,x,k+1}\|\boldsymbol{\mathcal{F}}(\boldsymbol{x}_{k, 1},\boldsymbol{u}_{k, 1})
- \boldsymbol{\mathcal{F}}(\boldsymbol{x}_{k, 2},\boldsymbol{u}_{k, 2})\|\\
&\leqslant L_{h,x,k+1} \Big( L_{\boldsymbol{\mathcal{F}},x}\|\boldsymbol{x}_{k, 1}-\boldsymbol{x}_{k, 2}\| + L_{\boldsymbol{\mathcal{F}},u}\|\boldsymbol{u}_{k, 1}-\boldsymbol{u}_{k, 2}\| \Big)\\
&\leqslant L_{h,x,k+1} \Big( L_{\boldsymbol{\mathcal{F}},x}\|\boldsymbol{x}_{k, 1}-\boldsymbol{x}_{k, 2}\| + L_{\boldsymbol{\mathcal{F}},u} L_{\pi}\|\boldsymbol{x}_{k, 1}-\boldsymbol{x}_{k, 2}\| \Big)\\
&= L_{h,x,k+1} (L_{\boldsymbol{\mathcal{F}},x} + L_{\boldsymbol{\mathcal{F}},u}L_{\pi}) \|\boldsymbol{x}_{k, 1}-\boldsymbol{x}_{k, 2}\|.
\end{aligned}
\]

For the second term, using the Lipschitz continuity of $h$ in the first argument, we obtain
\[
\bigl|h(\boldsymbol{x}_{k, 1},\mathbf{p}_{\mathbf{h}_k})-h(\boldsymbol{x}_{k, 2},\mathbf{p}_{\mathbf{h}_k})\bigr|
\leqslant L_{h,x,k}\|\boldsymbol{x}_{k, 1}-\boldsymbol{x}_{k, 2}\|.
\]

Combining the above inequalities yields
\[
|r_{k+1}(\boldsymbol{x}_{k,1})-r_{k+1}(\boldsymbol{x}_{k,2})|
\leqslant L \|\boldsymbol{x}_{k,1}-\boldsymbol{x}_{k,2}\|,
\]
where
\[
L := L_{h,x,k+1}(L_{\boldsymbol{\mathcal{F}},x} + L_{\boldsymbol{\mathcal{F}},u}L_{\pi})
      +(1-\eta)L_{h,x,k}.
\]
Thus $r_{k+1}$ is locally $L$-Lipschitz in $\boldsymbol{x}$ near
$\boldsymbol{x}_k^\star$, which completes the proof.
\end{proof}

Define the standard open Euclidean ball
$$
B_{\varepsilon_k}\left(\boldsymbol{x}_k^{\star(i)}\right):=\left\{\boldsymbol{x} \in \mathbb{R}^{d}:\left\|\boldsymbol{x}-\boldsymbol{x}_k^{\star(i)}\right\|_2<\varepsilon_k\right\},
$$
and we have the following proposition.

\begin{proposition}
\label{prop:local-feas-state}
Assume that $r_{k+1}$ is continuous in $\boldsymbol{x}_{k}$ in a neighborhood of an optimizer. Consider the optimization problem \eqref{eq:opt_problem}, which is only well-defined on the (implicit) domain
\[
\{\boldsymbol{x}_{k} : r_{k+1}(\boldsymbol{x}_{k}) > 0\}.
\]
As we have defined, $\boldsymbol{x}_k^{\star(i)}$ is the optimal solution to \eqref{eq:opt_problem} for sample initialization $i$ at time step $k$. By \eqref{eq:CBF-update}, we have strict interiority at the optimum, namely $r_{k+1}(\boldsymbol{x}_k^{\star(i)}) > 0$. Then the following conditions hold
\begin{enumerate}
\item \emph{Neighborhood feasibility:} There exists $\varepsilon_k > 0$ such that
\[
B_{\varepsilon_k}\big(\boldsymbol{x}_k^\star\big)
\ \subseteq\
\{\boldsymbol{x}_k : r_{k+1}(\boldsymbol{x}_k) > 0\},
\]
i.e., a (relative) open neighborhood of $\boldsymbol{x}_k^\star$ remains feasible.

\item There exists $\varepsilon_k' > 0$ such that for all
$\boldsymbol{x}_k$ with $\|\boldsymbol{x}_k - \boldsymbol{x}_k^{\star(i)}\| < \varepsilon_k'$,
\[
r_{k+1}(\boldsymbol{x}_k) \geqslant 0.
\]
\end{enumerate}
Moreover, as is proved in Lemma \ref{lem:r-lip}, $r_{k+1}$ is locally $L$-Lipschitz near $\boldsymbol{x}_k^{\star(i)}$. Therefore, any
\[
\begin{aligned}
\|\boldsymbol{x}_k - \boldsymbol{x}_k^{\star(i)}\| \leqslant \frac{r_{k+1}(\boldsymbol{x}_k^{\star(i)})}{L}
\Longrightarrow r_{k+1}(\boldsymbol{x}_k) \geqslant 0,
\end{aligned}
\]
yielding an explicit feasibility radius in the state space.
\end{proposition}

\begin{proof}
Let $\beta := r_{k+1}(\boldsymbol{x}_k^{\star(i)}) > 0$.
By continuity of $r_{k+1}$ at $\boldsymbol{x}_k^{\star(i)}$, there exists $\varepsilon_k > 0$ such that
\[
\|\boldsymbol{x}_k - \boldsymbol{x}_k^{\star(i)}\| < \varepsilon_k
\ \Longrightarrow\
\big|r_{k+1}(\boldsymbol{x}_k) - \beta\big| < \tfrac{\beta}{2}.
\]
Hence for all $\boldsymbol{x}_k$ in this neighborhood we have
\(
r_{k+1}(\boldsymbol{x}_k) > \beta/2 > 0,
\)
which proves item~(1) and naturally leads to item~(2) with some $\varepsilon_k' \leqslant \varepsilon_k$.

For the Lipschitz claim, by Lemma \ref{lem:r-lip} we have
\[
\big|r_{k+1}(\boldsymbol{x}_{k,1}) - r_{k+1}(\boldsymbol{x}_{k,2})\big|
\leqslant L\|\boldsymbol{x}_{k,1}-\boldsymbol{x}_{k,2}\|
\quad\text{locally}.
\]
Then
\[
\begin{aligned}
& \|\boldsymbol{x}_k - \boldsymbol{x}_k^{\star(i)}\|
\leqslant \frac{\beta}{L}
\ \Longrightarrow \\
& r_{k+1}(\boldsymbol{x}_k)
\geqslant r_{k+1}(\boldsymbol{x}_k^{\star(i)}) - L\cdot \tfrac{\beta}{L}
= 0,
\end{aligned}
\]
which establishes the explicit feasibility radius.
\end{proof}

\begin{remark}[On boundary optimizers in state space]
If an optimizer of \eqref{eq:opt_problem} happens to satisfy $r_{k+1}(\boldsymbol{x}_k^\star)=0$ (a boundary point), the neighborhood property above generally fails and arbitrarily small perturbations
of $\boldsymbol{x}_k^{\star(i)}$ may violate feasibility.
In practice, one enforces a positive safety margin by requiring a termination condition
$r_{k+1} \geqslant \beta > 0$ (``back-off''), or by appropriately scheduling the barrier parameter $\tau$.
\end{remark}

We then obtain the following corollary.

\begin{corollary}[Open neighborhood of feasible state sequences]
\label{corollary:open_neighborhood}
Let $\boldsymbol{X}^{\star(i)} := (\boldsymbol{x}_0^{\star(i)},\dots,\boldsymbol{x}_N^{\star(i)})$
be the optimal closed-loop state sequence by \eqref{eq:opt_problem} such that each $\boldsymbol{x}_k^{\star(i)}$, $0 \leqslant k \leqslant N$, is interior feasible in the sense of
Proposition~\ref{prop:local-feas-state}.
Then there exists a positive sequence
$\left(\varepsilon_0',\dots,\varepsilon_N'\right)$ such that the set
\[
\begin{aligned}
\boldsymbol{\mathfrak{X}}^{(i)}
:= \left \{ 
\boldsymbol{X}^{(i)}=(\boldsymbol{x}^{(i)}_0,\dots,\boldsymbol{x}^{(i)}_N) \middle|\right.\\
\left.\|\boldsymbol{x}^{(i)}_k - \boldsymbol{x}_k^{\star(i)}\| < \varepsilon_k', 0 \leqslant k \leqslant N
\right\}
\end{aligned}
\]
is a nonempty open subset of
$\mathbb{R}^{(N+1)n_x}$.
In particular, $\boldsymbol{\mathfrak{X}}^{(i)}$ is an open neighborhood of the nominal state sequence $\boldsymbol{X}^{\star(i)}$ in the trajectory space.
\end{corollary}

Proposition~\ref{prop:local-feas-state} guarantees that, at each time step $k$, there exists an open ball around the nominal state $\boldsymbol{x}_k^{\star(i)}$ on which the local safety condition $r_{k+1}(\boldsymbol{x}_k)\geqslant 0$ is satisfied. Taking the Cartesian product of these per-stage neighborhoods yields the open set $\boldsymbol{\mathfrak{X}}^{(i)}$ in Corollary~\ref{corollary:open_neighborhood}, which is an open neighborhood of the nominal closed-loop state sequence in $\mathbb{R}^{(N+1)n_x}$.

Concluding the previous results of this section, we have the following theorem.

\begin{theorem}
\label{thrm:46}
The set of reachable safe state sequences, namely $\boldsymbol{\mathcal{X}}_{r}$, can be covered by these finitely many open sets, i.e.,
\[
    \boldsymbol{\mathcal{X}}_{r} \subset \bigcup_{i=1}^{m} \boldsymbol{\mathfrak{X}}^{(i)},
\]
with $m$ being sufficiently large.
\end{theorem}

\begin{proof}
By construction, for every reachable safe state sequence $\boldsymbol{X}^{\star(i)} = (\boldsymbol{x}_0^{\star(i)},\dots,\boldsymbol{x}_N^{\star(i)}) \in \boldsymbol{\mathcal{X}}_{r}$, there exists a nonempty open neighborhood $\boldsymbol{\mathfrak{X}}(\boldsymbol{X}^{\star(i)})$ such that $\boldsymbol{X}^{\star(i)} \in \boldsymbol{\mathfrak{X}}(\boldsymbol{X}^{\star(i)}) \subset \boldsymbol{\mathcal{X}}_{r}$.
Hence, the family $\{\boldsymbol{\mathfrak{X}}(\boldsymbol{X}^{\star(i)})\}_{\boldsymbol{X}^{\star(i)} \in \boldsymbol{\mathcal{X}}_{r}}$ forms an open cover of $\boldsymbol{\mathcal{X}}_{r}$. Since $\boldsymbol{\mathcal{X}}_{r}$ is compact, the Heine--Borel theorem (see e.g. \cite{rudin1976principles}) implies that this open cover admits a finite subcover. Therefore, there exist trajectories
\[
    \boldsymbol{X}^{\star(1)},\dots,\boldsymbol{X}^{\star(m_0)} \in \boldsymbol{\mathcal{X}}_{r}
\]
such that
\[
    \boldsymbol{\mathcal{X}}_{r}
    \subset \bigcup_{i=1}^{m_0} \boldsymbol{\mathfrak{X}}\bigl(\boldsymbol{X}^{\star(i)}\bigr).
\]
For any $m \geqslant m_0$, we may add further sampled trajectories and corresponding open sets without affecting this inclusion, so the statement holds for all sufficiently large $m$.
\end{proof}

\begin{remark}
Theorem~\ref{thrm:46} shows that, under the proposed algorithm, the inferred system states and control inputs remain feasible for any test initial condition $\boldsymbol{x}(0) \in \mathcal{X}_{r}(0)$, provided that the neural control policy is trained on a finite number of data samples.

Although several efforts have been made to empirically improve the feasibility of DPCs from a probabilistic perspective, the results in this section constitute, to the best of our knowledge, the first rigorous attempt to prove strict feasibility of both state and control outputs generated by neural controllers. In particular, we show that strict feasibility of the learning-based controller can be ensured using only a finite number of training samples, which appears to be the first result of this kind in the literature on white-box learning-based control.
\end{remark}

In addition, we could also prove the following theorem.

\begin{theorem}
\label{thrm:47}
    With the increase of the number of training data samples $m$, the probability of inferred controls and the corresponding states to violate the feasibility conditions is non-increasing and converges to $0$ with a finite $m$.
\end{theorem}

\begin{proof}

Define the $m$-dependent subset of \emph{guaranteed feasible} closed-loop state sequences
\[
    \boldsymbol{\mathcal{X}}_{\mathrm{feas}}^{(m)}
    := \bigcup_{i=1}^{m} \boldsymbol{\mathfrak{X}}^{(i)} \subset \boldsymbol{\mathcal{X}}_{r}.
\]
By construction (local feasibility around each nominal trajectory and continuity of the closed-loop dynamics and constraints), for every $\boldsymbol{X} \in \boldsymbol{\mathcal{X}}_{\mathrm{feas}}^{(m)}$ the inferred control inputs and the corresponding state sequence satisfy all feasibility constraints. Hence any constraint violation can only occur for state sequences in the complement
\[
    \boldsymbol{\mathcal{V}}^{(m)}
    := \boldsymbol{\mathcal{X}}_{r} \setminus \boldsymbol{\mathcal{X}}_{\mathrm{feas}}^{(m)},
\]
where ``$\setminus$'' denotes the set-difference operator. Let $\chi$ be a probability measure on $\boldsymbol{\mathcal{X}}_{r}$ describing the distribution of (reference) closed-loop state sequences induced by random initial conditions $\boldsymbol{x}_0 \in \mathcal{X}_{r}(0)$. 
Define
\[
    p_m
    := \mathbb{P}\{\text{violation of feasibility with $m$ training samples}\}.
\]
From the above discussion, any violation must correspond to a trajectory in $\boldsymbol{\mathcal{V}}^{(m)}$, and therefore
\begin{equation}
    p_m 
    \leqslant \chi\bigl(\boldsymbol{\mathcal{V}}^{(m)}\bigr)
    = \chi\bigl(\boldsymbol{\mathcal{X}}_{r} \setminus \boldsymbol{\mathcal{X}}_{\mathrm{feas}}^{(m)}\bigr).
    \label{eq:pm-upper-bound}
\end{equation}

As $m$ increases, we only add more open sets to the union, hence
\[
    \boldsymbol{\mathcal{X}}_{\mathrm{feas}}^{(m)}
    = \bigcup_{i=1}^{m} \boldsymbol{\mathfrak{X}}^{(i)}
    \subseteq
    \bigcup_{i=1}^{m+1} \boldsymbol{\mathfrak{X}}^{(i)}
    = \boldsymbol{\mathcal{X}}_{\mathrm{feas}}^{(m+1)}.
\]
Therefore the complements form a nested decreasing sequence,
\[
    \boldsymbol{\mathcal{V}}^{(m+1)}
    = \boldsymbol{\mathcal{X}}_{r} \setminus \boldsymbol{\mathcal{X}}_{\mathrm{feas}}^{(m+1)}
    \subseteq
    \boldsymbol{\mathcal{X}}_{r} \setminus \boldsymbol{\mathcal{X}}_{\mathrm{feas}}^{(m)}
    = \boldsymbol{\mathcal{V}}^{(m)}.
\]
Since $\chi$ is a probability measure, this implies
\[
    \chi\bigl(\boldsymbol{\mathcal{V}}^{(m+1)}\bigr)
    \le
    \chi\bigl(\boldsymbol{\mathcal{V}}^{(m)}\bigr),
\]
and by \eqref{eq:pm-upper-bound} we obtain
\[
    p_{m+1} \leqslant p_{m}, \qquad \forall\,m \in \mathbb{N}.
\]
Hence the violation probability $p_m$ is non-increasing in $m$.

By Proposition~\ref{thrm:46}, the compact set $\boldsymbol{\mathcal{X}}_{r}$ can be covered by finitely many of these open neighborhoods, i.e., there exists an integer $m_0 \in \mathbb{N}$ and trajectories
\[
    \boldsymbol{X}^{\star(1)},\dots,\boldsymbol{X}^{\star(m_0)} \in \boldsymbol{\mathcal{X}}_{r}
\]
such that
\[
    \boldsymbol{\mathcal{X}}_{r}
    \subset \bigcup_{i=1}^{m_0} \boldsymbol{\mathfrak{X}}\bigl(\boldsymbol{X}^{\star(i)}\bigr).
\]
If the training set contains at least these $m_0$ trajectories (for example, by actively designing the training data to realize this finite subcover), then for any $m \geqslant m_0$ we have
\[
    \boldsymbol{\mathcal{X}}_{\mathrm{feas}}^{(m)}
    \supset \bigcup_{i=1}^{m_0} \boldsymbol{\mathfrak{X}}\bigl(\boldsymbol{X}^{\star(i)}\bigr)
    \supset \boldsymbol{\mathcal{X}}_{r},
\]
and therefore
\[
    \boldsymbol{\mathcal{V}}^{(m)}
    = \boldsymbol{\mathcal{X}}_{r} \setminus \boldsymbol{\mathcal{X}}_{\mathrm{feas}}^{(m)}
    = \emptyset,
    \qquad m \geqslant m_0.
\]
Consequently, $\chi\bigl(\boldsymbol{\mathcal{V}}^{(m)}\bigr)=0$ for all $m \geqslant m_0$, and by \eqref{eq:pm-upper-bound} we obtain
\[
    p_m = 0, \qquad \forall\, m \geqslant m_0.
\]

Combining the two steps, we have shown that the sequence $\{p_m\}_{m\in\mathbb{N}}$ is non-increasing and, moreover, becomes identically zero for all $m$ larger than some finite threshold $m_0$. 
Equivalently, the probability that the inferred controls and the corresponding state trajectories violate the feasibility conditions is non-increasing in $m$ and converges to zero in a finite number of steps.
\end{proof}

\begin{remark}
    Theorem~\ref{thrm:47} shows that each training sample (assumed to satisfy the CBF condition) in the optimization problem \eqref{eq:opt_problem} contributes positively to the learned DPC policy in a topological sense. In most existing DPC results, this property is largely taken for granted and justified only empirically; Theorem~\ref{thrm:47} provides a formal proof of this intuitive idea. At present, the theorem is an existence result, and it motivates us to seek constructive proofs that can directly guide the training of DPC, its variants, and other white-box learning-based controllers from a topological and geometric perspective inspired by Theorem~\ref{thrm:47}. 
\end{remark}

One would naturally think about how to determine the number of data samples required for training, in the aim of guaranteeing the exact feasibility. In the following part of this section, we will investigate into this problem and propose preliminary solutions.

\begin{lemma}[Explicit feasibility radius]
\label{lemma:explicit-radius}
Let Assumptions~\ref{assmp:41}--\ref{assmp:43} hold, and let
$\boldsymbol{X}^{\star(i)}$ be a training trajectory satisfying the strict
CBF condition
$ r_{k+1}(\boldsymbol{x}_k^{\star(i)})>0$, $k \in \{0, \dots, N-1\}$.
Suppose that each $r_{k+1}$ is Lipschitz continuous on a neighborhood of
$\boldsymbol{x}_k^{\star(i)}$ with constant $L_k>0$, and define the uniform
Lipschitz constant
$
L:=\max_{k=0,\ldots,N-1} L_k .
 $
Then, for every
\begin{equation} \label{epsilonki}
\boldsymbol{x}_k \in B_{\epsilon_k^{(i)}}(\boldsymbol{x}_k^{\star(i)}),
\qquad
\epsilon_k^{(i)}
=
\frac{r_{k+1}(\boldsymbol{x}_k^{\star(i)})}{L},
\end{equation}
provided this ball is contained in the corresponding Lipschitz neighborhood,
the one-step CBF condition satisfies
$
r_{k+1}(\boldsymbol{x}_k) \geqslant 0 .
 $
Consequently, the learned policy preserves one-step state feasibility within the Euclidean ball
$B_{\epsilon_k^{(i)}}(\boldsymbol{x}_k^{\star(i)})$ with radius $\epsilon_k^{(i)}$.
\end{lemma}

\begin{proof}
By Lipschitz continuity of $r_{k+1}$, for any
$\boldsymbol{x}_k$ in the considered neighborhood,
\[
|r_{k+1}(\boldsymbol{x}_k)
-r_{k+1}(\boldsymbol{x}_k^{\star(i)})|
\le
L_k\|\boldsymbol{x}_k-\boldsymbol{x}_k^{\star(i)}\|.
\]
Since $L \geqslant L_k$, the radius
$
\epsilon_k^{(i)}
=
\frac{r_{k+1}(\boldsymbol{x}_k^{\star(i)})}{L}
 $
is a conservative lower bound on the local feasibility radius. Hence, if
 $
\|\boldsymbol{x}_k-\boldsymbol{x}_k^{\star(i)}\|
\le
\epsilon_k^{(i)},
 $
then
$ 
L_k\|\boldsymbol{x}_k-\boldsymbol{x}_k^{\star(i)}\|
\le
L\|\boldsymbol{x}_k-\boldsymbol{x}_k^{\star(i)}\|
\le
r_{k+1}(\boldsymbol{x}_k^{\star(i)}).
$
Therefore,
\[
r_{k+1}(\boldsymbol{x}_k)
\ge
r_{k+1}(\boldsymbol{x}_k^{\star(i)})
-
L_k\|\boldsymbol{x}_k-\boldsymbol{x}_k^{\star(i)}\|
\geqslant 0 .
\]
Thus, every state in
$B_{\epsilon_k^{(i)}}(\boldsymbol{x}_k^{\star(i)})$
satisfies the one-step CBF condition under the learned policy
\[
B_{\epsilon_k^{(i)}}(\boldsymbol{x}_k^{\star(i)})
\subseteq
\{\boldsymbol{x}_k:r_{k+1}(\boldsymbol{x}_k)\geqslant 0\}.
\]
If, in addition, the ball is restricted to the reachable set at time $k$, then
\[
B_{\epsilon_k^{(i)}}(\boldsymbol{x}_k^{\star(i)})\cap \mathcal R(k)
\subseteq
\mathcal X_r(k).
\]
\end{proof}

\begin{proposition}[Maximal required sample size]
\label{prop:maximal-sample-size}
Given that there exists a set of $m_{0}$ closed-loop trajectories in the training set, which satisfy that the feasibility balls $B_{\epsilon_{k}^{(i)}}\left(\boldsymbol{x}_{k}^{\star(i)}\right)$ (with conservative radii $\epsilon_{k}^{(i)}$ defined in \eqref{epsilonki}) cover the reachable safe set $\mathcal{X}_{r}(k)$ for each time step $k$. Then, $m_{0}$ constitutes the maximal number of training samples necessarily required to guarantee exact feasibility.
\end{proposition}

\begin{proof}
The proof follows from the compactness of $\mathcal{X}_r(k)$ and the conservative nature of the feasibility radius derived in Lemma \ref{lemma:explicit-radius}.

By assumption, the union of the $m_0$ feasibility balls (constructed using the derived radii) fully covers the safe set
\begin{equation}
    \mathcal{X}_r(k) \subseteq \bigcup_{i=1}^{m_0} B_{\epsilon_{k}^{(i)}}\left(\boldsymbol{x}_{k}^{\star(i)}\right).
\end{equation}

Crucially, the feasibility radius $\epsilon_{k}^{(i)}$ in \eqref{epsilonki} is computed using the uniform Lipschitz coefficient $L$, which represents the global worst-case sensitivity. Therefore, $\epsilon_{k}^{(i)}$ serves as a conservative lower bound (minimum value) for the true localized feasibility radius. The actual feasibility region induced by a sample $\boldsymbol{x}_{k}^{\star(i)}$ is geometrically larger than or equal to the conservative ball $B_{\epsilon_{k}^{(i)}}$. This leads to the following inclusion
\begin{equation}
    B_{\epsilon_{k}^{(i)}}\left(\boldsymbol{x}_{k}^{\star(i)}\right) \subseteq B_{\bar{\epsilon}_{k}^{(i)}}\left(\boldsymbol{x}_{k}^{\star(i)}\right),
\end{equation}
where $\bar{\epsilon}_{k}^{(i)}$ denotes the radii of the true feasible ball around the sample.

Consequently, the union of the actual feasibility regions forms a superset of the conservative cover
\begin{equation}
    \bigcup_{i=1}^{m_0} B_{\bar{\epsilon}_{k}^{(i)}}\left(\boldsymbol{x}_{k}^{\star(i)}\right) \supseteq \bigcup_{i=1}^{m_0} B_{\epsilon_{k}^{(i)}}\left(\boldsymbol{x}_{k}^{\star(i)}\right) \supseteq \mathcal{X}_r(k).
\end{equation}

This inclusion confirms that $m_0$ samples are strictly sufficient. Moreover, since $m_0$ is determined based on the conservative (minimum) radii, the number of samples required under the actual (larger) feasibility regions would effectively be smaller. Thus, $m_0$ serves as a conservative upper bound, representing the maximal number of samples necessarily required for the feasibility guarantee.
\end{proof}

In Figure \ref{Illustration}, we give a visualization and conclusion of the core ideas of the previous theoretical results.

\begin{figure}[htbp]
\centering
\includegraphics[scale=0.11]{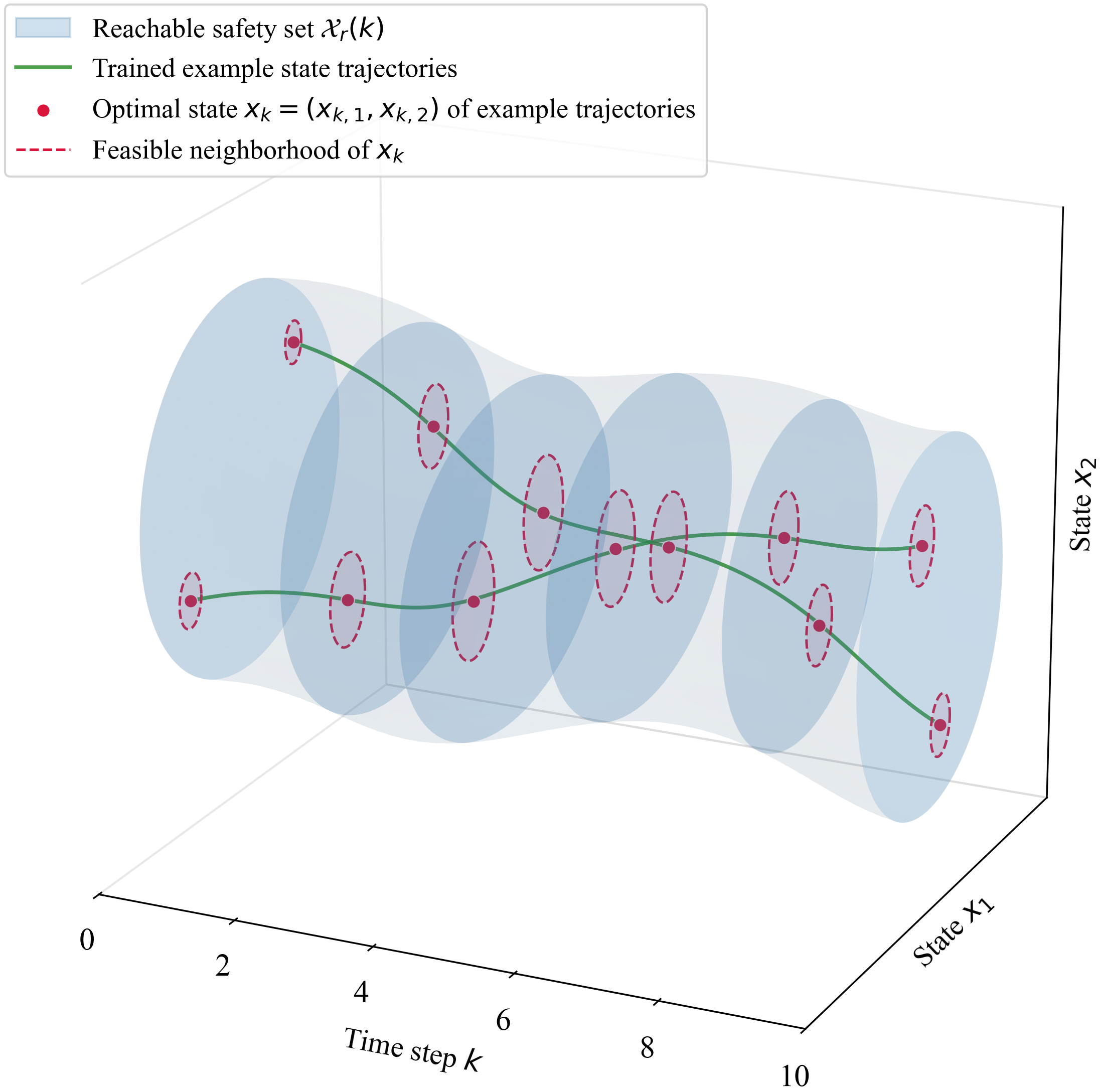}
\caption{Illustration of the feasibility analysis for DPC in a two-dimensional state space. Each training trajectory generated by self-supervised learning of DPC induces a strictly feasible local neighborhood around the nominal state sequence. The size of this neighborhood is determined by the feasibility margin $r_{k+1}(\boldsymbol{x}_k^{\star(i)})$ and the local Lipschitz constant $L$, yielding the conservative feasibility radius $\epsilon_k^{(i)}= r_{k+1}(\boldsymbol{x}_k^{\star(i)}) / L$. As additional training trajectories are introduced, these local feasible regions collectively form a finite cover of the reachable safe set $\mathcal{X}_r(k)$, implying that exact feasibility can be guaranteed through offline training without requiring online optimization.}
\label{Illustration}
\end{figure}

\section{CBF-proxy DPC via a Warm-Started Homotopy Training Scheme for Feasibility Guarantee}

Leveraging the differentiable architecture of the DPC, the previous analyses derived rigorous feasibility guarantees predicated on the assumption that the CBF update condition \eqref{eq:CBF-update} is strictly satisfied across the training dataset. Although optimizing \eqref{Lnc} via soft penalties (e.g. ReLU \cite{drgovna2024learning}) implicitly guides the closed-loop trajectories toward the safe set, guaranteeing exact topological feasibility requires the explicit mathematical enforcement of \eqref{eq:CBF-update}. To systematically ensure strict interiority with respect to the feasibility constraints, we propose an offline CBF-based proxy filtering scheme. The resulting synthesized controller is referred to as the CBF-proxy DPC.

Moreover, as we have mentioned in the preceding sections, directly training a neural controller with strict barrier functions (e.g., logarithmic barriers) is notoriously challenging due to numerical brittleness when initialized outside the safe set, due to which soft penalties are mostly used for feasibility. To overcome this, we propose a two-stage \emph{homotopy} training scheme. This scheme combines the broad optimization shaping of soft penalties (e.g., ReLU and GeLU) during a warm-start phase with the strict enforcement of hard constraints via an offline control barrier function (CBF) safety proxy during fine-tuning.

\subsubsection*{Stage 1: Warm-start via soft penalties}
In the first stage, we train the policy network $\mathbf{W}$ by minimizing the standard DPC loss $\mathcal{L}_{NC}$ defined in \eqref{Lnc}. Here, the constraint violations are penalized using a soft hinge-style function. This allows the neural controller to efficiently explore the state space and quickly converge to a nominally feasible and optimal trajectory manifold, avoiding the singularities of barrier methods. 

Let the state sequence generated by this warm-started neural controller be denoted as the nominal state sequence $\breve{\boldsymbol{x}}^{(i)}_{k}$, for each sample initialization $i$.

\subsubsection*{Stage 2: Offline CBF-proxy fine-tuning}
To strictly enforce feasibility without relying on computationally expensive online optimization (e.g., online safety filters via quadratic programming (QP)) and provide feasibility guarantees in an offline manner for the system trajectory of each sample $i$, we introduce a subsequent offline fine-tuning stage. 

For a given control input $\boldsymbol{u}_{k}^{(i)}$ generated by the neural controller for sample $i$, we define the CBF residual augmented with a non-negative slack variable $\delta_{k,l}^{(i)} \geqslant 0$:
\[
\begin{aligned}
r_{k+1, l}(\boldsymbol{x}^{(i)}_k)
:=& h_{l}\!\left(\boldsymbol{\mathcal{F}}(\boldsymbol{x}^{(i)}_k, \boldsymbol{u}^{(i)}_k),
        \mathbf{p}_{\mathbf{h}_{k+1}}\right) \\
& - (1-\eta)
h_{l}\left(\boldsymbol{x}^{(i)}_{k}, \mathbf{p}_{\mathbf{h}_{k}}\right)+\delta_{k,l}^{(i)}.
\end{aligned}
\]
The slack $\delta_{k,l}^{(i)}$ softens the strict condition \eqref{eq:CBF-update} during early iterations of fine-tuning, preserving numerical feasibility. When \eqref{eq:CBF-update} is fully satisfied, $\delta_{k,l}^{(i)}=0$; otherwise, it quantifies the temporary dynamic violation. 

To embed this mechanism seamlessly into the offline neural network training, we formulate the \emph{CBF-proxy loss}:
\begin{equation}
\label{eq:safety-proxy-loss}
\begin{aligned}
& \mathcal{L}_{CBF}(\boldsymbol{x}_{k}, \boldsymbol{u}_k, \mathbf{r}_k)\\
:= & \sum_{i=1}^m \sum_{k=0}^{N-1} \bigg( \|\boldsymbol{x}^{(i)}_k(\mathbf{W}) - \breve{\boldsymbol{x}}^{(i)}_k\|_2^2
+ \rho\left(\delta_{k,l}^{(i)}\right)^{2}\\
- & \tau \sum_{l=1}^{n_{h}}\log\!\big(r_{k+1, l}(\boldsymbol{x}^{(i)}_k(\mathbf{W})_{k,l}^{(i)})\big) \bigg).
\end{aligned}
\end{equation}

This loss acts as a differentiable, offline analog to a safety filter QP. The first term tracks the optimal nominal trajectory $\breve{\boldsymbol{x}}^{(i)}_k$ obtained from Stage 1. The weight $\rho>0$ heavily penalizes the slack, making $\delta_{k,l}^{(i)}>0$ a last resort. Finally, instead of modifying the CBF itself, the log-barrier term is applied strictly to the residual $r_{k+1, l}$. This biases the learned policy toward the strict interior of the control invariant set over time, alleviating the brittleness of pure penalty methods.

The network parameters $\mathbf{W}$ and the auxiliary slacks are then fine-tuned via the following optimization performed entirely offline
\begin{equation}
\label{eq:safety-filter-opt}
\begin{aligned}
\min_{\mathbf{W},\ \delta_{k,l}^{(i)}\ge0}
& \mathcal{L}_{NC}(\boldsymbol{x}_{k}, \boldsymbol{u}_k, \mathbf{r}_k)+\mathcal{L}_{CBF}(\boldsymbol{x}_{k}, \boldsymbol{u}_k, \mathbf{r}_k),
\end{aligned}
\end{equation}
subject to the system dynamics and operational constraints \eqref{eq:opt_problem_dyn} to \eqref{eq:opt_problem_vars}. By performing Stage 2,  $\delta_{k,l}^{(i)} = 0$ serves as a certificate of feasibility satisfaction for the $l_{\text{th}}$ constraint on system states, considering sample $i$ at time step $k$.

In contrast to the state constraints, the feasibility of the control inputs can be guaranteed \textit{a prior} by the network architecture, which is a less challenging task. Specifically, strict satisfaction of the input bounds is enforced by applying a saturating activation function paired with an affine transformation (e.g., a normalized hyperbolic tangent mapping) at the final layer of the neural policy, thereby mathematically confining the generated control inputs to the admissible region.

Based on the previous results of this section, we are now able to construct a DPC with CBF-proxy loss for ODE, of which an algorithm is given in Algorithm \ref{alg:1}. CBF-proxy DPC successfully bridges the gap between soft neural training and strict deterministic safety. It is worth noting that, for the proposed CBF-proxy DPC scheme, we do not require any online optimizer to adjust the inferred control inputs. Such online ``safety filtering'' mechanisms are typically introduced to enforce the feasibility of the inferred control actions, but they may undermine the main advantages of learning-based controllers, namely their fast inference and trajectory-wise optimality. 

\begin{algorithm}
    \caption{Training of CBF-proxy DPC}
    \label{alg:1}
    \begin{algorithmic}[1]
        \Require The prediction horizon $N$; the system dynamics \eqref{eq:opt_problem_dyn}; the feasible state and control conditions \eqref{eq:opt_problem_state}, \eqref{eq:opt_problem_input}; maximum training epochs $E_1$ and $E_2$.
        \State Sample $m$ initial system states $\boldsymbol{x}_{0}^{(i)}, 1 \leqslant i \leqslant m$
        
        \Statex \textbf{Stage 1: Warm-start via soft penalties}
        \State $e \gets 0$
        \While{$e < E_1$}
            \State Update the policy weights $\mathbf{W}$ for one epoch by minimizing the cost function $\mathcal{L}_{NC}(\boldsymbol{x}_{k}, \boldsymbol{u}_k, \mathbf{r}_k)$ \eqref{Lnc}
            \State $e \gets e + 1$
        \EndWhile
        
        \Statex \textbf{Stage 2: CBF-proxy fine-tuning}
        \State $e \gets 0$
        \While{$e < E_2$}
            \State Update the policy weights $\mathbf{W}$ for one epoch by minimizing the cost function $\mathcal{L}_{NC}(\boldsymbol{x}_{k}, \boldsymbol{u}_k, \mathbf{r}_k) + \mathcal{L}_{CBF}(\boldsymbol{x}_{k}, \boldsymbol{u}_k, \mathbf{r}_k)$ \eqref{eq:safety-proxy-loss}
            \State Check the satisfaction of feasibility $l$, namely $\delta_{k,l}^{(i)}$ for sample $i$ at time step $k$
            \If{$\delta_{k,l}^{(i)} = 0$ for all $i, k, l$}
                \State \textbf{break} \Comment{Early stopping if strict safety constraints are satisfied}
            \EndIf
            \State $e \gets e + 1$
        \EndWhile
    \end{algorithmic}
\end{algorithm}

We would also like to emphasize that the offline training with CBF-proxy loss serves two purposes. First, it improves the feasibility of the control actions inferred at the training samples. Second, it allows us to identify, for each data sample $i$ and constraint $l$ at time step $k$, whether feasibility is ensured, i.e., whether $\delta_{k, l}^{i*} = 0$. We then have the following proposition.

\begin{proposition}[Deterministic closed-loop feasibility]
\label{thrm:deterministic_feasibility}
Let Assumptions~\ref{assmp:41}--\ref{assmp:43} hold. Let $\mathcal{D}_m = \{\boldsymbol{X}^{\star(i)}\}_{i=1}^m$ be a finite dataset of $m$ training trajectories generated by Algorithm~\ref{alg:1}. Suppose that after the offline CBF-proxy fine-tuning (Stage 2), the optimal slack variables satisfy $\delta_{k,l}^{(i)*} = 0, \forall i \in \{1,\dots,m\}, \forall k \in \{0,\dots,N-1\}, \ \forall l \in \{1,\dots,n_h\}$, thereby ensuring strict one-step interior feasibility $r_{k+1,l}(\boldsymbol{x}_k^{\star(i)}) > 0$. 
If the training dataset $\mathcal{D}_m$ is sufficiently dense such that the induced open feasibility balls cover the reachable safe set at each time step $k$, i.e.,
\begin{equation*}
    \mathcal{X}_r(k) \subseteq \bigcup_{i=1}^{m} B_{\epsilon_k^{(i)}}\left(\boldsymbol{x}_k^{\star(i)}\right), \quad \forall k \in \{0, \dots, N\}
\end{equation*}
where the conservative local feasibility radius is given by $\epsilon_k^{(i)} = \frac{r_{k+1}(\boldsymbol{x}_k^{\star(i)})}{L}$ with the uniform Lipschitz constant $L$ defined in \eqref{LipschitzCoeff}, then the learned DPC policy $\pi_{\mathbf{W}}(\boldsymbol{\mathrm{x}}_k, \boldsymbol{\xi}_k)$ is recursively feasible by design. 
\end{proposition}

Consequently, for any initial state $\boldsymbol{\mathrm{x}}_0 \in \mathcal{X}_r(0)$, every closed-loop trajectory generated by the explicit policy satisfies the state and input constraints deterministically for all execution steps, namely, $h_l(\boldsymbol{\mathrm{x}}_k, \mathbf{p}_{\mathbf{h}_k}) \geqslant 0 \quad \text{and} \quad g(\boldsymbol{\mathrm{u}}_k, \mathbf{p}_{\mathbf{g}_k}) \geqslant 0, \quad \forall k \geqslant 0$, without requiring online optimization or \textit{a posteriori} safety filtering.

\begin{remark}
    Well-distributed samples are essential to the proposed framework. Related work includes recent feasible-set sampling methods for MPC, such as the LP-based Hit-and-Run approach~\cite{Milios2026}. Establishing explicit theoretical guarantees, including sample-complexity bounds and covering-number analyses, remains an important direction for future research.
\end{remark}

\section{Exact Feasibility by DPC: A Perspective}

While the theoretical results presented in the preceding sections primarily serve as existence proofs, they unveil a fundamental shift in the certification of learning-based control systems. In the existing literature, safe learning-based control predominantly addresses feasibility through a probabilistic lens, often relying on concentration inequalities, statistical confidence bounds, or chance-constrained formulations. Such methodologies conventionally treat the system dynamics merely as a generative environment for empirical validation, yielding guarantees that hold only with high probability. In contrast, the proposed CBF-proxy DPC framework exploits the explicit white-box structure intrinsic to DPC, effectively transitioning the validation paradigm from stochastic approximation to deterministic topological certification. By proving that strict feasibility can be ensured uniformly across all admissible trajectories using only a finite number of training samples, we establish a rigorous deterministic foundation that is largely absent in contemporary learning-based control.

This deterministic guarantee necessitates a philosophical reinterpretation of the training dataset. In standard supervised or reinforcement learning paradigms, an increase in data volume primarily correlates with a reduction in expected empirical risk. However, under the proposed topological framework, data samples acquire a rigorous geometric interpretation. As demonstrated in Section 3, each training trajectory that strictly satisfies the feasibility condition induces a localized safety certificate, mathematically formalized as an open Euclidean ball within the state space. By invoking the compactness of the reachable safe set and the Heine-Borel theorem, we establish that a finite subcover of these localized certificates is sufficient to tile the entire solution manifold. Consequently, the concept of sample efficiency is redefined: the algorithmic objective is no longer to asymptotically approximate an underlying probability distribution, but rather to identify a finite set of geometric anchors that deterministically cover a compact space.

Central to this topological certification is the fundamental dual role of the system dynamics. Black-box reinforcement learning algorithms typically restrict the role of the environment to that of a forward simulator, providing state transitions and reward signals while obscuring the underlying mathematical structure. Conversely, DPC's differentiable architecture grants direct analytical access to the system Jacobians, thereby facilitating a precise derivation of the solution manifold's Lipschitz properties. This mathematical transparency allows for the explicit quantification of the localized feasibility radii and, by extension, the maximal necessary sample complexity required to ensure strict feasibility. Thus, the system dynamics transcend their conventional utility as a mere trajectory generator; their intrinsic structural properties are actively harvested to provide a deterministic safety certification that remains fundamentally inaccessible to model-free methodologies.

Although the current findings establish a critical existence proof---demonstrating that the probability of constraint violation converges to zero given a finite number of suitably distributed training samples---they simultaneously illuminate a constructive pathway for future algorithmic design. Recognizing that the conservative feasibility radii furnish a strict upper bound on the necessary sample size, future research must pivot toward topology- and geometry-guided active sampling. Rather than relying on uniform or random data generation, tailored sampling algorithms could be synthesized for specific classes of system dynamics. Such algorithms would strategically identify and prioritize training samples capable of covering maximal subsets of the reachable space, thereby achieving exact feasibility with an optimally minimal dataset.

Ultimately, the profound significance of DPC extends far beyond its well-documented computational advantages over implicit MPC in real-time applications. By seamlessly integrating an explicit, white-box dynamical model within a differentiable neural policy framework, DPC uniquely bridges the critical chasm between rapid learning-based inference and rigorous, deterministic safety certification. It elevates the neural controller from a statistically approximate heuristic to a topologically certifiable operator. Consequently, this paradigm not only substantiates the deployment of learning-based predictive control in stringent, safety-critical environments but also establishes a foundational theoretical intersection where finite data sampling, system geometry, and classical control theory converge to yield exact, uncompromised feasibility.

\section{Numerical results}

In this section, we apply the proposed CBF-proxy DPC algorithm to three problems for validating our theoretical observations in the previous sections and the performance of our proposed algorithms.

\subsection{Nonholonomic mobile robot navigation}

We first propose to treat a nonholonomic mobile robot navigation problem. The system represents a standard unicycle model navigating through a 2D grid environment populated with static obstacles. Its kinematics are described by the following system of ordinary differential equations
\begin{equation} 
\begin{aligned} 
\dot{x}_t &= v_t \cos(\theta_t) \\
\dot{y}_t &= v_t \sin(\theta_t) \\
\dot{\theta}_t &= \omega_t.
\end{aligned} 
\end{equation} 

Here $\boldsymbol{x}_t = [x_t, y_t, \theta_t]^\top \in \mathbb{R}^3$ is the state vector representing the 2D spatial coordinates and heading angle of the robot, $\boldsymbol{u}_t = [v_t, \omega_t]^\top \in \mathbb{R}^2$ is the control input comprising the linear and angular velocities, respectively. In each case, a random map is generated, and the robot must navigate from a start position $(0.5,0.5)$ with an initial heading of $\pi / 4$ to the goal $(9.5,9.5)$ over a horizon of $T=120$ time steps $(\Delta t=0.4s)$. The safety parameters are selected as $\eta=0.1, \rho=50$, and $\tau=0.05$.

For digital implementation, this continuous-time model is discretized. The deterministic drift component is integrated using a forward Euler method with a sampling time of $\Delta t = 0.4$\,s, yielding a nominal discrete-time map $\boldsymbol{\mathcal{F}}(\boldsymbol{x}_k, \boldsymbol{u}_k)$. The prediction and control horizon is set to $N = 120$ steps. The operating environment is configured as a $10 \times 10$ spatial grid map, where an obstacle ratio of $10\%$ is maintained by randomly assigning static obstacles. The side length of each grid is $1$.

The training process of the CBF-proxy DPC is structured into two distinct stages. We first train the base DPC policy for $800$ epochs to establish fundamental goal-reaching and obstacle-avoidance behaviors. This stage minimizes a composite hinge-loss cost function penalizing distance to the target, boundary violations, obstacle collisions, and control effort. Upon convergence of this nominal control policy, we further fine-tune the network for $200$ epochs using the DPC loss together with the safety-proxy loss to strictly improve the safety feasibility of the trajectories. Specifically, we enforce a discrete-time CBF residual condition with $\eta = 0.1$. During this second stage, the optimization explicitly penalizes the slack variable with a quadratic weight of $\rho = 50.0$ and applies a logarithmic barrier function with a weight of $\tau = 0.05$ to push the residual $r_{k+1}$ away from zero.

The control policy is parameterized by a multi-layer perceptron (MLP) featuring an encoder with three hidden layers of $512$, $256$, and $128$ neurons, utilizing GeLU activation functions. The network branches into two output heads to simultaneously predict the control sequence $\boldsymbol{u}$ and the relaxation variables $\boldsymbol{\delta}$. To train the policy, we construct a dataset of $N_s = 100, 200, 400, 600, 800, 1000$ randomly generated environments. The parameters of the neural policy are optimized using the Adam optimizer with a learning rate of $10^{-3}$ and a batch size of $64$. The key benefit of this formulation is that it is fully differentiable, enabling the use of automatic differentiation over the unrolled unicycle dynamics to directly compute the policy gradients. Finally, to rigorously evaluate the safety guarantees, the trained policy is subjected to a statistical performance test across $1000$ unseen environment configurations to quantify the empirical collision rate.

\begin{figure}[htbp]
\centering
\includegraphics[scale=0.2]{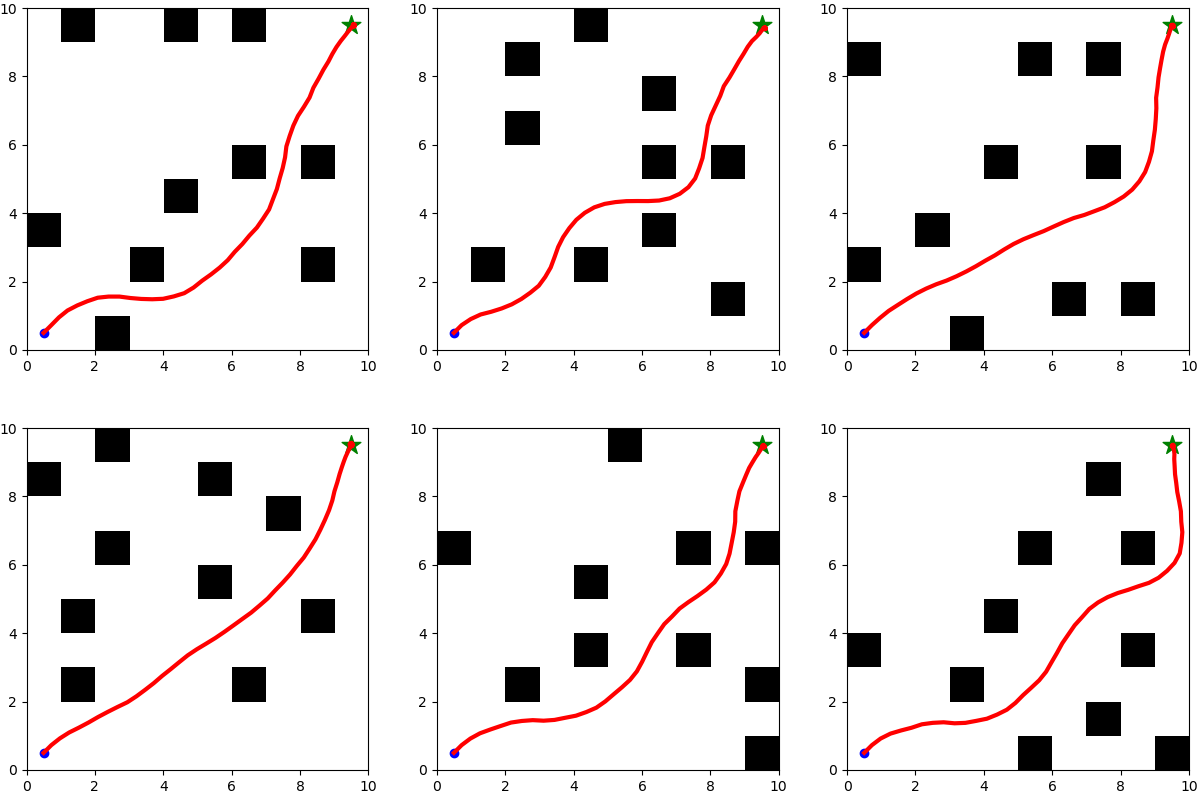}
\centering
\caption{Six sample trajectories of the mobile robot navigation inferred by the proposed CBF-proxy DPC algorithm. }
\label{fig2}
\end{figure}

For evaluation, we conduct statistical tests over $1000$ independent test environment cases. At every time step along these trajectories, we check whether the state violates the constraints (namely violates the obstacle constraints) and then report the pointwise violation ratio over all state samples and the fraction of trajectories that exhibit at least one violation. Resulting curves for the pointwise violation ratios and trajectory violation ratios versus the number of training data samples are plotted in Figure \ref{fig3} and Figure \ref{fig4} respectively. We note that both curves decrease monotonically with the increase of the number of the training data samples, which converge to zero when $1000$ samples are used for training. These results are consistent with theoretical results in Theorem~\ref{thrm:47}, providing empirical validation of the theorem.

\begin{figure}[htbp]
\centering
\includegraphics[scale=0.25]{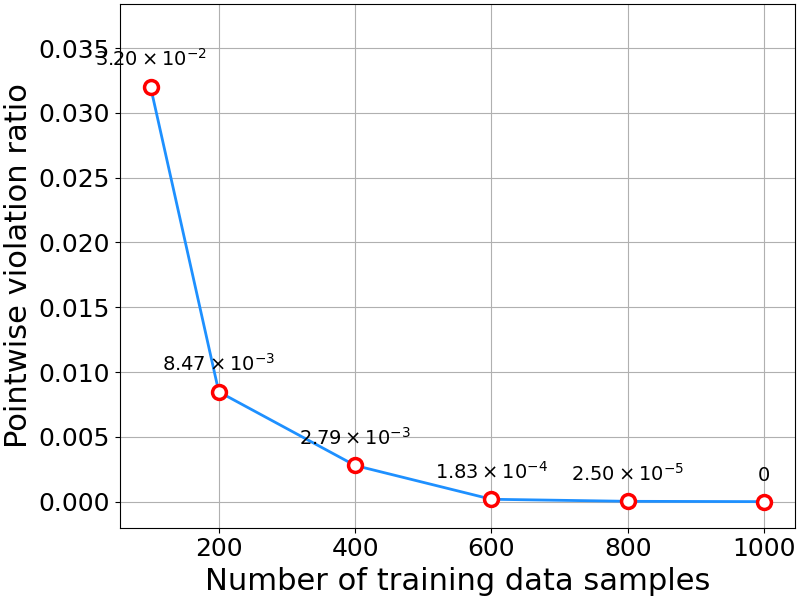}
\centering
\caption{The figure of total violation points versus the number of training data samples. The ratio decreases monotonically with the increase of the training data samples.}
\label{fig3}
\end{figure}

\begin{figure}[htbp]
\centering
\includegraphics[scale=0.25]{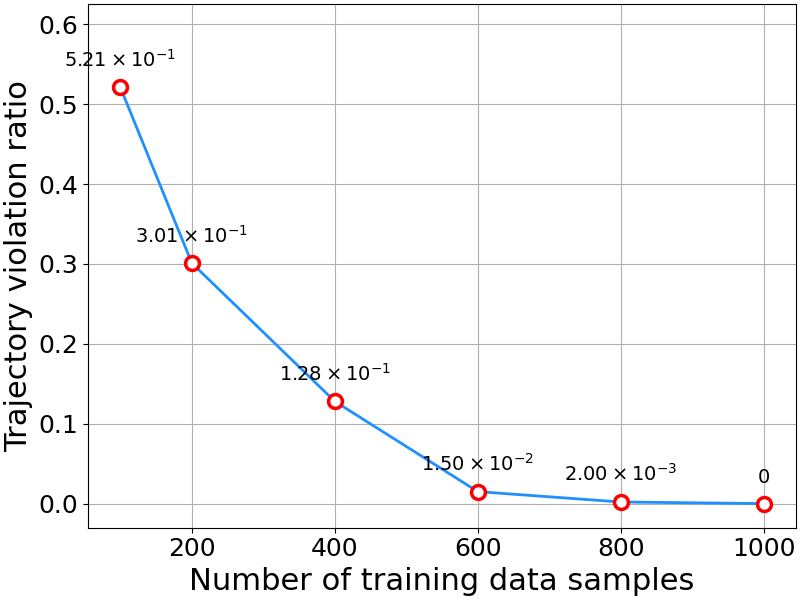}
\centering
\caption{The figure of trajectory violation ratio versus the number of training data samples.}
\label{fig4}
\end{figure}

\subsection{Unstable linear system stabilization}
To evaluate the proposed algorithm on a more challenging regulation task, we consider the stabilization of an open-loop unstable continuous-time linear system. Specifically, we investigate a single-input underactuated system governed by the following dynamics
\begin{equation}
    \dot{x}(t) = A x(t) + B u(t)
\end{equation}
where $x(t) \in \mathbb{R}^{2}$ is the state vector and $u(t) \in \mathbb{R}$ is the scalar control input. The entries of both the state matrix $A \in \mathbb{R}^{2 \times 2}$ and the input matrix $B \in \mathbb{R}^{2 \times 1}$ are randomly sampled from a standard normal distribution. To strictly enforce open-loop instability, the randomly initialized state matrix $A$ is subsequently shifted, if necessary, to ensure that the real part of its maximum eigenvalue is strictly greater than 0.2. Meanwhile, the randomly generated input matrix $B$, which maps the scalar control input into the state derivative space, yields a non-trivial, coupled underactuated configuration. To guarantee exact reproducibility of these randomized system dynamics, a fixed random seed is employed throughout the numerical experiments. For digital implementation and offline training, the continuous-time dynamics are discretized using a fourth-order Runge-Kutta (RK4) integrator with a sampling time of 0.05 s. 

The control objective is to stabilize the system to the origin while strictly adhering to hard physical constraints. Specifically, the actuator control input is strictly bounded within $u_k \in [-20, 20]$, and the safe feasible region is defined as a bounded hypercube in the state space, $x_k \in [-16, 16]^{2}$.

The neural control policy is parameterized by an MLP architecture featuring two hidden layers with 64 neurons each and GeLU activation functions. To structurally guarantee that the inferred control actions never violate the actuator limits, the output layer of the network is explicitly bounded. The neural policy is trained offline using a dataset of $50$, $75$, $100$, $200$, $300$ initial conditions sampled uniformly from a localized region of $[-2, 2]^{2}$ respectively. The optimization objective is formulated as a composite loss function comprising a standard quadratic regulation cost with a state weight of 10 and a control weight of 0.1, alongside a heavy terminal penalty of 100 to enforce asymptotic convergence. Furthermore, to embed safety specifications during the unrolled trajectory optimization, a ReLU-based soft penalty is formulated; this penalty activates whenever the predicted states approach within a predefined margin of 1.0 from the hard state boundaries.

Training a neural policy over long prediction horizons for unstable dynamics typically suffers from the exploding gradient problem. To systematically mitigate this and ensure strict safety, the training procedure employs a two-stage strategy. In the warm-up training stage, the network focuses on driving the system to the origin without the safety-proxy loss active. This stage uses a curriculum learning approach where the prediction horizon is progressively expanded from a short horizon of $10$ steps up to $200$ steps. The network parameters are updated using the AdamW optimizer, with the learning rate decaying from $5 \times 10^{-3}$ to $1 \times 10^{-3}$ as the horizon extends. In the second stage, the network is trained by both DPC and safety-proxy losses for $100$ epochs over a fixed maximum prediction horizon of $200$ steps. During this phase, the safety-proxy loss is activated to push trajectories safely away from the hard boundaries, and the learning rate is further reduced to $5 \times 10^{-4}$. This two-stage mechanism enables the neural controller to first learn fundamental stabilization maneuvers before mastering long-term constraint satisfaction and asymptotic regulation.

To empirically validate the topological feasibility and stabilization guarantees, closed-loop Monte-Carlo simulations were conducted post-training. The learned policies were evaluated on $1000$ independent, randomly generated initial conditions drawn from the initial state region. Each scenario was simulated for $1000$ time steps (with a sampling time of $T_s = 0.05s$), corresponding to $50s$ of closed-loop operation. Figure \ref{fig5} and Figure \ref{fig7} show the phase portraits of $1000$ randomly initialized test samples with the CBF-proxy DPC trained by $50$ and $100$ training samples respectively. Figure \ref{fig6} and Figure \ref{fig8} show the state trajectories and control inputs corresponding to $50$ and $100$ training samples. As illustrated in Figure \ref{fig9}, the constraint violation rate among test trajectories decreases monotonically, ultimately converging to zero as the training dataset size reaches $100$. Furthermore, the proposed CBF-proxy DPC exhibits a significantly lower violation rate compared to the standard DPC baseline.

The numerical results confirm that the proposed learning-based controllers succeeds in stabilizing the divergent linear dynamics. Crucially, the empirical evaluation reveals that the constraint violation rate monotonically decreases as the number of training samples increases. As observed, while training with $50$ samples yields a $96.4\%$ convergence rate with severe divergence in uncovered regions in the reachable set, expanding the dataset size effectively mitigates these edge cases, actively corroborating our theoretical proofs. Furthermore, the ablation study results shown in Figure \ref{fig9} justify the integration of the safety-proxy loss. Comparative results demonstrate that fine-tuning the neural policy with this safety-proxy loss essentially decreases the ratio of violating the constraints. These findings validate both the deterministic feasibility of the inferred neural control sequences under sufficient data regimes and the critical role of the safety-proxy loss in ensuring strict constraint satisfaction.
\begin{figure}[htbp]
\centering
\includegraphics[scale=0.3]{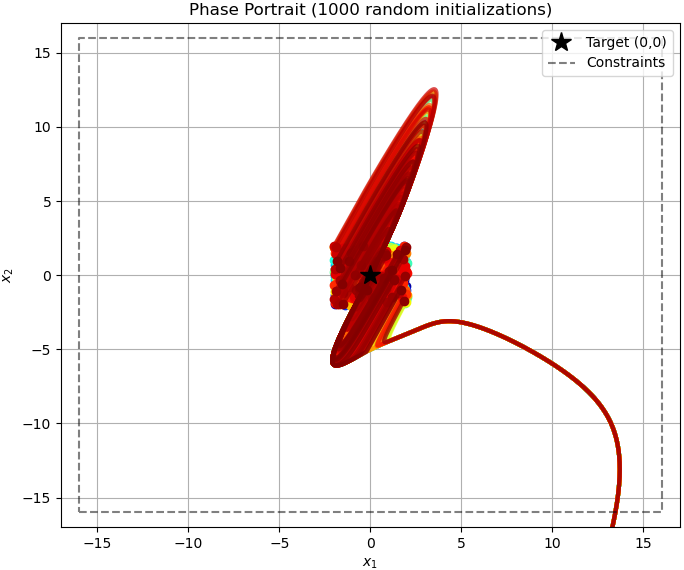}
\centering
\caption{Phase trajectories of $1000$ samples in the testing dataset, which are initialized randomly. The CBF-proxy DPC is trained by $50$ data samples. By the proposed algorithm, $96.4\%$ of the testing samples are successfully stabilized to the origin.}
\label{fig5}
\end{figure}

\begin{figure}[htbp]
\centering
\includegraphics[scale=0.3]{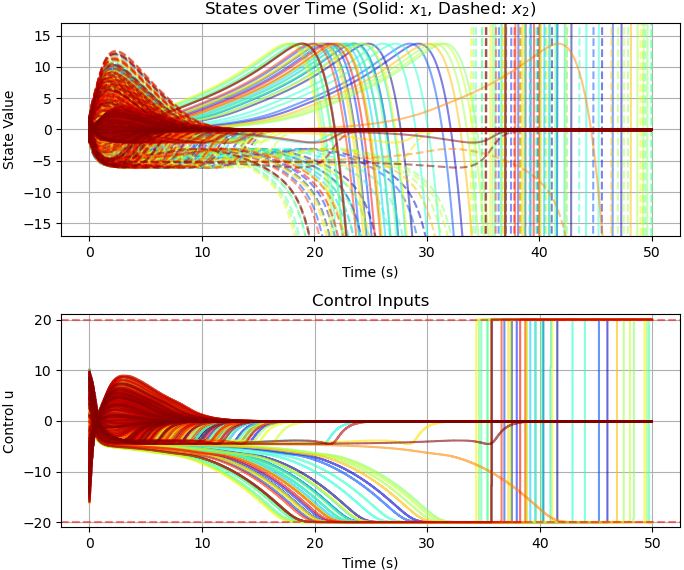}
\centering
\caption{State trajectories and control inputs of test data using $50$ training samples. While $96.4\%$ of the test trajectories successfully converge to the origin, the remaining cases exhibit severe divergence. This indicates that certain regions of the reachable set remain uncovered by the feasible neighborhoods of the training data.}
\label{fig6}
\end{figure}

\begin{figure}[htbp]
\centering
\includegraphics[scale=0.3]{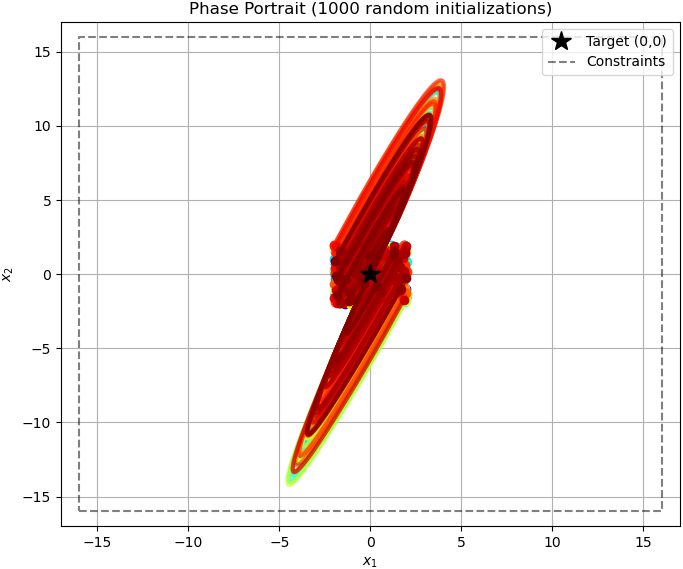}
\centering
\caption{Phase trajectories of $1000$ samples in the testing dataset, which are initialized randomly. The CBF-proxy DPC is trained by $100$ data samples. By the proposed algorithm, all testing samples are successfully stabilized to the origin.}
\label{fig7}
\end{figure}

\begin{figure}[htbp]
\centering
\includegraphics[scale=0.3]{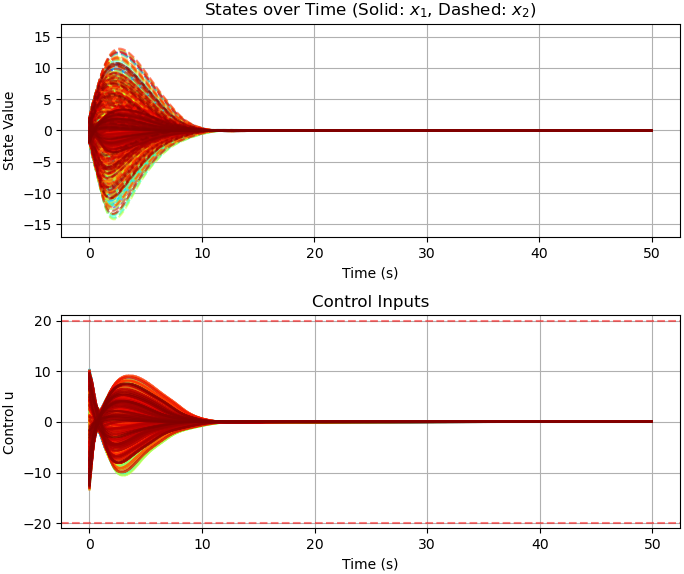}
\centering
\caption{State trajectories and control inputs of test data using $100$ training samples, all of which converge to the origin.}
\label{fig8}
\end{figure}

\begin{figure}[htbp]
\centering
\includegraphics[scale=0.25]{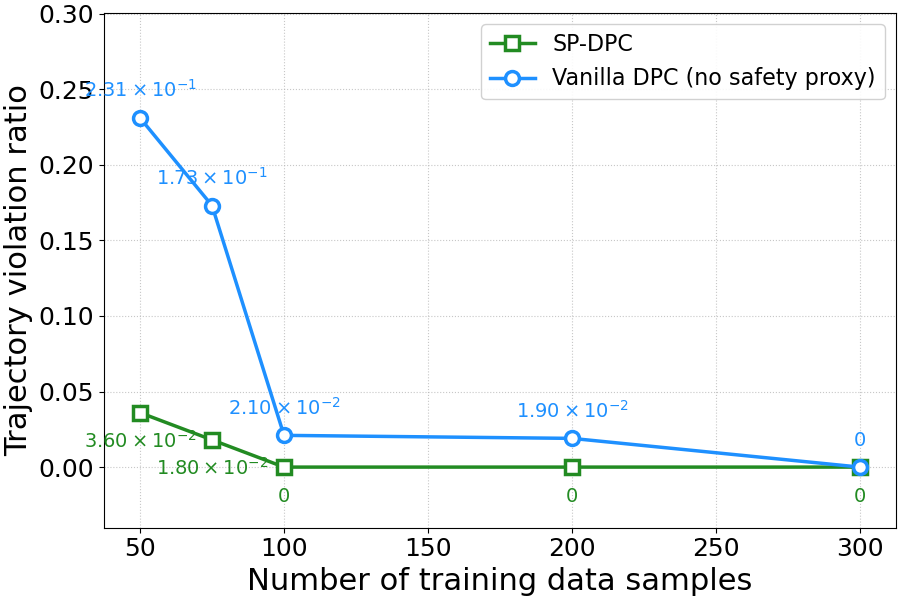}
\centering
\caption{The ratio of test data violating constraints versus the number of training data samples for CBF-proxy DPC and vanilla DPC, respectively. We note that the violation ratio is significantly smaller for the former one.}
\label{fig9}
\end{figure}

\subsection{Constrained quadcopter system stabilization}
In the following part of this section, we simulate our algorithm on a 12-dimensional nonlinear quadcopter system, a benchmark widely recognized for its highly coupled, underactuated dynamics and stringent safety requirements. The system state is defined as $x = [p_x, p_y, p_z, \phi, \theta, \psi, \dot{p}_x, \dot{p}_y, \dot{p}_z, \dot{\phi}, \dot{\theta}, \dot{\psi}]^T \in \mathbb{R}^{12}$, comprising the 3D position, Euler angles (roll, pitch, yaw), and their respective linear and angular velocities. The control input $u \in \mathbb{R}^4$ corresponds to the individual thrust forces generated by the four rotors. Because the quadcopter has six degrees of freedom but only four actuators, it is severely underactuated; any translational motion in the horizontal plane ($X$-$Y$) is intrinsically coupled with the rotational dynamics (roll and pitch), making constraint satisfaction uniquely challenging.
While the physical quadcopter is governed by complex nonlinear equations (involving trigonometric couplings and aerodynamic drag), we employ its widely-used discrete-time linearized model around the hovering equilibrium point (with a sampling time $T_s = 0.1$ s) for the synthesis and evaluation of the predictive controllers. This formulation not only provides a standardized framework for fair performance benchmarking but also mirrors standard practices in real-time model predictive control.

The primary control objective is to stabilize the quadcopter at a predefined hovering altitude (e.g., $p_z = 1.0$ m, with zero tilt and zero velocity) from various initial conditions, while strictly adhering to hard state constraints and actuator saturation. To maintain the validity of the small-angle approximation inherent in the linearized model, strict attitude constraints are imposed on the roll and pitch angles ($\phi, \theta \in [-\pi/6, \pi/6]$), alongside boundaries for the $X$ and $Z$ axes. Furthermore, the control inputs are bounded by physical motor saturation limits, $u \in [u_{\min}, u_{\max}]$, defined relative to the nominal hovering thrust. Synthesizing a neural policy that simultaneously guarantees asymptotic tracking and strict constraint satisfaction under these coupled conditions is notoriously challenging, particularly in data-scarce regimes.

To empirically validate the data efficiency and safety robustness of the proposed framework, we benchmark our CBF-proxy DPC against a standard DPC baseline. The baseline is trained with a conventional soft-penalty loss function, which superimposes a quadratic (LQR-type) state tracking and control effort cost with ReLU-based exterior penalty terms. We provide the phase portrait of position $X$ and altitude $Z$ for $200$ samples in the total $2000$ with $16$ training samples, in Figure \ref{fig10}. In Figure \ref{fig12}, the closed-loop performance is quantitatively evaluated using the trajectory violation ratio, defined as the proportion of simulated trajectories that breach any state or input constraints during the operational horizon. To rigorously assess the sample efficiency and generalization capability of the neural policies, this metric is evaluated across severely constrained training datasets, ranging from $N = 6$ to $N = 16$ initial state samples. The monotonic decrease of the trajectory violation ratio with the increase of the number of training data samples also align with our theoretical observations.

\begin{figure}[htbp]
\centering
\includegraphics[scale=0.3]{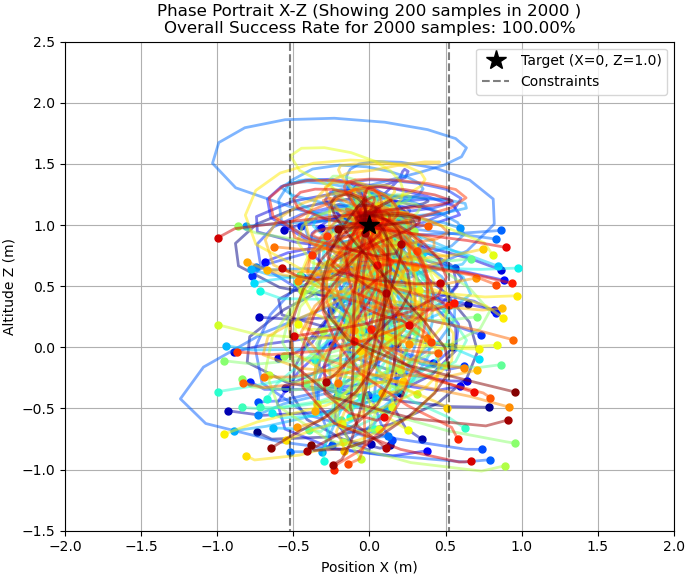}
\centering
\caption{Phase trajectories of $200$ in $2000$ testing data samples, which are initialized randomly. The CBF-proxy DPC is trained by $16$ data samples. By the proposed algorithm, all testing samples are successfully stabilized to the desired state.}
\label{fig10}
\end{figure}

\begin{figure}[htbp]
\centering
\includegraphics[scale=0.3]{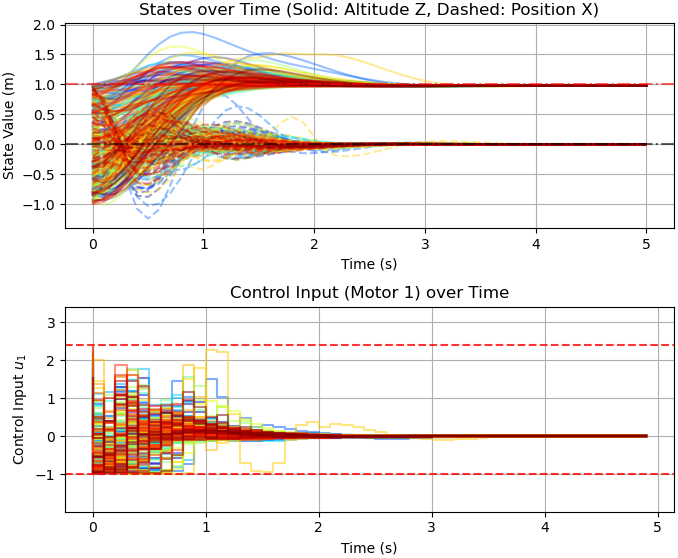}
\centering
\caption{State trajectories (Altitude $Z$ and Position $X$) and control inputs (Control Input $u_{1}$) of test data using $16$ training samples, which show successful stabilization for the $2000$ testing samples. }
\label{fig11}
\end{figure}

\begin{figure}[htbp]
\centering
\includegraphics[scale=0.25]{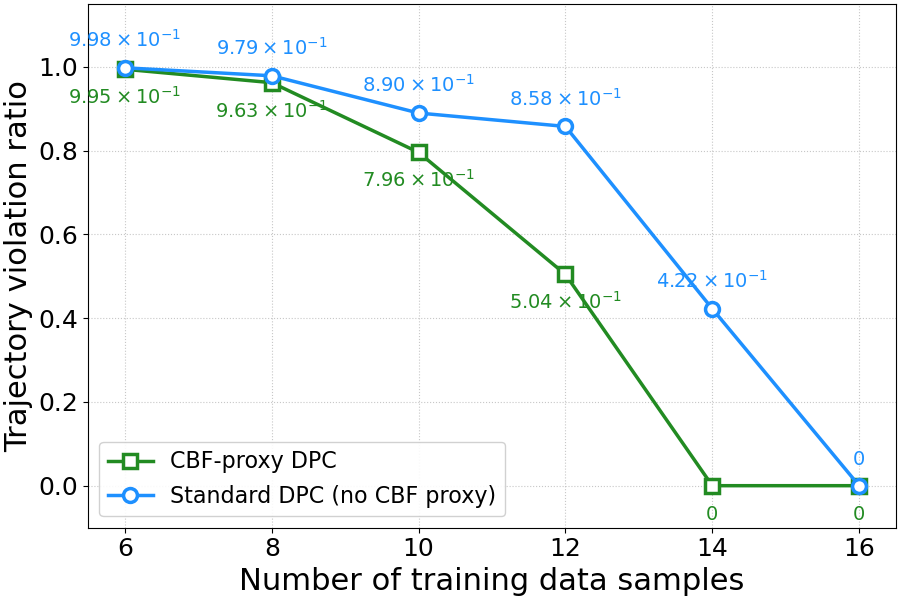}
\centering
\caption{The ratio of test data violating constraints versus the number of training data samples for CBF-proxy DPC and standard DPC, respectively. We note that the violation ratio is significantly smaller for the former one.}
\label{fig12}
\end{figure}

\section{Conclusions}

Learning-based control methods have demonstrated promising results across various engineering domains, offering significantly faster online computation times than implicit MPC and making them highly attractive for real-time implementation. However, the widespread deployment of these learning mechanisms in safety-critical systems remains severely hindered by the absence of rigorous feasibility guarantees. In this paper, we established a formal feasibility theory for DPC by leveraging a novel topological analysis of the induced reachable safe set. To realize this theoretical framework, we introduced an offline policy learning strategy equipped with a dedicated safety-proxy structure designed to strictly enforce constraint satisfaction over the training data.

Crucially, we rigorously proved that the inferred neural control sequences deterministically satisfy all feasibility conditions, provided a finite number of training samples are perfectly optimized to yield zero safety residuals offline. Furthermore, we analytically demonstrated that as the size of this strictly safe training dataset increases, the probability of constraint violation by the inferred policy monotonically decreases, ultimately converging to zero within a finite sample limit. By deeply exploiting the white-box nature of DPC, this work provides, to the best of our knowledge, the first rigorous and deterministic feasibility guarantees for such architectures from a topological perspective, offering a definitive structural advantage over conventional black-box methods.

The theoretical findings in this paper establish fundamental existence-type guarantees regarding topological feasibility and sample complexity. Transitioning from these existence proofs to fully constructive bounds is a natural progression of this research rather than a fundamental limitation. In future work, we intend to leverage deeper geometric and topological insights into specific classes of system dynamics to develop constructive algorithms that explicitly enforce these theoretical feasibility limits \textit{a priori} when training DPC offline.

\bibliographystyle{ieeetr}
\bibliography{Reference}

\end{document}